\documentclass[aip,graphicx,jmp]{revtex4-1}
\usepackage{amsmath,amssymb,amsthm}
\usepackage[T2A]{fontenc}

\draft 

\newtheorem{proposition}{Proposition}
\newtheorem{corollary}{Corollary}

\newtheorem{remark}{Remark}
\newtheorem{theorem}{Theorem}

\newtheorem{lemma}{Lemma}
\newtheorem{Example}{Example}
\begin{document}

\title{A new bidirectional generalization of (2+1)-dimensional matrix k-constrained KP hierarchy} 

\author{O.I. Chvartatskyi}
\email{alex.chvartatskyy@gmail.com}
\author{Yu.M. Sydorenko}
\email{y_sydorenko@franko.lviv.ua}
\affiliation{Ivan Franko National University of L'viv}

\date{\today}

\begin{abstract}
We introduce a new bidirectional generalization of (2+1)-dimensional
k-constrained KP hierarchy ((2+1)-BDk-cKPH). This new hierarchy
generalizes (2+1)-dimensional k-cKP hierarchy, $(t_A,\tau_B)$ and
$(\gamma_A,\sigma_B)$ matrix hierarchies. (2+1)-BDk-cKPH contains a
new matrix (1+1)-k-constrained KP hierarchy. Some members of
(2+1)-BDk-cKPH are also listed. In particular, it contains matrix
generalizations of DS systems, (2+1)-dimensional modified
Korteweg-de Vries equation and the Nizhnik equation. (2+1)-BDk-cKPH
also includes new matrix (2+1)-dimensional generalizations of the
Yajima-Oikawa and Melnikov systems. Binary Darboux Transformation
Dressing Method is also proposed for construction of exact solutions
for equations from (2+1)-BDk-cKPH. As an example the exact form of
multi-soliton solutions for vector generalization of the
Davey-Stewartson system is given.

\end{abstract}

\pacs{05.45 Yv, 05.45-a, 02.30 Jr, 02.10 Ud}

\maketitle 

\section{Introduction}\label{dfdf}
In the modern theory of nonlinear integrable systems, algebraic
methods play an important role (see the survey in \cite{LDA}). Among
them there are the Zakharov-Shabat dressing method
\cite{Zakharov,Zakh-Manak,solitons}, Marchenko's method
\cite{March}, and an approach based on the Darboux-Crum-Matveev
transformations \cite{Matveev79,Matveev}. Algebraic methods allow us
to omit analytical difficulties that arise in the investigation of
corresponding direct and inverse scattering problems for nonlinear
equations. A significant contribution to such methods has also been
made by the Kioto group \cite{DJKM1,DJKM2,SS3,MM4,Ohta}. In
particular, they investigated scalar and matrix hierarchies for
nonlinear integrable systems of Kadomtsev-Petviashvili type (KP
hierarchy).

The KP hierarchy 
is of fundamental importance in the theory of integrable systems and
shows up in various ways in mathematical physics. Several extensions
and generalizations of it have been obtained. For example, the
multi-component KP
hierarchy contains several physically relevant nonlinear integrable
systems, including the Davey-Stewartson equation, the
two-dimensional Toda lattice and the three-wave resonant interaction
system. There are several equivalent formulations of this hierarchy:
matrix pseudo-differential operator (Sato) formulation,
$\tau$-function approach via matrix Hirota bilinear identities,
multi-component free fermion formulation.
Another kind of generalization is the so-called ``KP equation with
self-consistent sources'' (KPSCS), discovered by Melnikov
\cite{M1,Mf,M4,M2,M3}. In \cite{SS,KSS,Chenga1,CY,Chenga2},
k-symmetry constraints of the KP hierarchy which have connections
with KPSCS were investigated. 
The resulting k-constrained KP (k-cKP) hierarchy contains physically
relevant systems like the nonlinear Schr\"odinger equation, the
Yajima-Oikawa system, a generalization of the Boussinesq equation,
and the Melnikov system. Multi-component generalizations of the
k-cKP hierarchy were introduced in \cite{SSq}. In the papers
\cite{ZC,Oevel93,Oevel96,Aratyn97} the differential type of the
gauge transformation operator was applied to the k-constrained KP
hierarchy at first. In \cite{Chau1} differential and integral type
of the gauge transformation operators were applied to k-cKP
hierarchy. In \cite{WLG} the binary Darboux transformations for
dressing of the multi-component generalizations of the k-cKP
hierarchy were investigated. A modified k-constrained KP (k-cmKP)
hierarchy was proposed in \cite{CY,KSO,OC}. It contains, for
example, the vector Chen-Lee-Liu and the modified KdV (mKdV)
equation. Multi-component versions of the Kundu-Eckhaus and
Gerdjikov-Ivanov \cite{Gerdjikov1,Gerdjikov2} equations were also
obtained in \cite{KSO}, via gauge transformations of the k-cKP,
respectively the k-cmKP hierarchy.

Moreover, in \cite{MSS,6SSS}, (2+1)-dimensional extensions of the
k-cKP hierarchy were introduced and dressing methods via
differential transformations were investigated. In \cite{BS1,PHD},
exact solutions for some representatives of the (2+1)-dimensional
k-cKP hierarchy were obtained by dressing binary Darboux
transformations. This hierarchy was also rediscovered recently in
\cite{LZL1}. Dressing methods via differential transformations for
this hierarchy and its modified version were investigated in
\cite{LZL2}. The (2+1)-dimensional k-cKP hierarchy in particular
contains the DS-III ($k=1$), Yajima-Oikawa ($k=2$) and Melnikov
($k=3$) hierarchies. The corresponding Lax representations of this
hierarchy (see (\ref{spa1})) consists of one differential and one
integro-differential operator. Our aim was to generalize the
(2+1)-dimensional k-cKP hierarchy (\ref{spa1}) to the case of two
integro-differential operators in Lax pair. It is essential to call
this new hierarchy a bidirectional generalization of
(2+1)-dimensional k-cKP hierarchy or simply (2+1)-BDk-cKP hierarchy
(see (\ref{ex2+1})). We will consider this hierarchy in the most
general matrix case. Another aim of our investigation is a
construction of the Binary Darboux Transformation Dressing Method
for (2+1)-BDk-cKP hierarchy (\ref{ex2+1}).

 This work is organized as follows. In Section 2 we
 present a short survey of results on constraints for KP hierarchies
 and their (2+1)-dimensional generalizations. In Section \ref{extended} we
 introduce (2+1)-BDk-cKP hierarchy. Members of the
 obtained hierarchy are also listed there.  (2+1)-BDk-cKP hierarchy contains matrix generalizations of
  Davey-Stewartson hierarchy (first members of it are two different matrix versions of DS-III system:
  DS-III-a (\ref{DSc2=0}) and DS-III-b (\ref{DSc1=0})), new Yajima-Oikawa and Melnikov hierarchies.
  In Section \ref{dressed} we consider
 dressing via binary Darboux transformation for the hierarchy constructed in Section
 \ref{extended}. Exact forms of solutions for some members of this hierarchy  are
 also presented here in terms of Grammians.
 In the final section,
we discuss the obtained results and mention problems for further
investigations. Some examples of Lax representations from the
bidirectional generalization of (2+1)-dimensional k-constrained
modified KP hierarchy ((2+1)-BDk-cmKP hierarchy) are also presented
there.

\section{$k$-constrained KP hierarchy and its extensions}\label{kckp}

To make this paper self-contained, we briefly introduce the KP
hierarchy \cite{LDA}, its k-symmetry constraints (k-cKP hierarchy),
and the extension of the k-cKP hierarchy to the (2+1)-dimensional
case \cite{MSS,6SSS}. A Lax representation of the KP hierarchy is
given by
\begin{equation}\label{ssk}
L_{t_n}=[B_n,L],\,    \qquad n\geq1,
\end{equation}
where $L=D+U_1D^{-1}+U_2D^{-2}+\ldots$ is a scalar
pseudodifferential operator, $t_1:=x$, $D:=\frac{\partial}{\partial
x}$, and $B_n:= (L^n)_+
 := (L^n)_{\geq0}=D^n+\sum_{i=0}^{n-2}u_iD^i$ is
the differential operator part of $L^n$. The consistency condition
(zero-curvature equations), arising from the commutativity of flows
(\ref{ssk}), is
\begin{equation}
B_{n,t_k}-B_{k,t_n}+[B_n,B_k]=0.
\end{equation}
Let $B^{\tau}_n$ denote the formal transpose of $B_n$, i.e.
$B^{\tau}_n:=(-1)^nD^n+\sum_{i=0}^{n-2}(-1)^iD^iu^{\top}_i$, where
$^{\top}$ denotes the matrix transpose. We will use curly brackets
to denote the action of an operator on a function whereas, for
example, $B_n \, q$ means the composition of the operator $B_n$ and
the operator of multiplication by the function $q$. The following
formula holds for $B_nq$ and $B_n\{{q}\}$:
$B_n\{{q}\}=B_nq-(B_n{q})_{>0}.$ The k-cKP hierarchy
\cite{SS,KSS,Chenga1,CY,Chenga2} is given by
\begin{equation}\label{eq1}
  L_{t_n}=[B_n,L],
  \end{equation}
with the k-symmetry reduction
\begin{equation}\label{eq2}
    L_k:=L^k=B_k+{q}D^{-1}{r}.
\end{equation}
The hierarchy given by (\ref{eq1})-(\ref{eq2}) admits the Lax
representation (here $k\in{\mathbb{N}}$ is fixed):
\begin{equation}\label{Laxeq}
  [L_k,M_n]=0,\,\,\, L_k=B_k+qD^{-1}r,\,\, M_n=\partial_{t_n}-B_n.
\end{equation}
Lax equation (\ref{Laxeq}) is equivalent to the following system:
\begin{equation}\label{sys1}
[L_k,M_n]_{\geq0}=0,\,\, M_n\{q\}=0,\,\,\,M_n^{\tau}\{{r}\}=0.
\end{equation}
Below we will also use the formal adjoint
$B^*_n:=\bar{B}^{\tau}_n=(-1)^nD^n+\sum_{i=0}^{n-2}(-1)^iD^i{u}^*_i$
of $B_n$, where $^\ast$ denotes the Hermitian conjugation (complex
conjugation and transpose).
 In \cite{SSq}, multi-component (vector) generalizations of the
 k-cKP hierarchy were introduced,
\begin{equation}\label{operLk}
L_k:=L^k=B_k+\sum_{i=1}^m\sum_{j=1}^mq_im_{ij}D^{-1}r_j=B_k+{\bf
q}{\cal M}_0D^{-1}{\bf
r}^{\top}, 
\end{equation}
where ${\bf q}=(q_1,\ldots,q_m)$ and ${\bf r}=(r_1,\ldots,r_m)$ are
vector functions, ${\cal M}_0=(m_{ij})_{i,j=1}^m$ is a constant
$m\times m$ matrix. ${\cal M}_0$ can be annihilated by the change of
functions: ${\bf q}{\cal M}_0\rightarrow{\bf q}$, ${\bf
r}^{\top}\rightarrow{\bf r}^{\top}$ (or ${\bf q}\rightarrow{\bf q}$,
${\cal M}_0{\bf r}^{\top}\rightarrow{\bf r}^{\top}$). However, we
remain this matrix because it plays an important role in Section
\ref{dressed} where we consider dressing methods for (2+1)-BDk-cKP
hierarchy (\ref{ex2+1}).
We shall note that extensions of hierarchies given by (\ref{eq1})
and (\ref{operLk}), namely (2+1)-dimensional k-cKP (\ref{LMcom}) and
(2+1)-BDk-cKP (\ref{frex}), (\ref{frexxx}) hierarchies are
represented analogously to (\ref{sys1}) (with more general $L_k$ and
$M_n$ operators).

 For $k=1$, the hierarchy given by (\ref{eq1}) and (\ref{operLk}) is a multi-component generalization of the AKNS
hierarchy. For $k=2$ and $k=3$, one obtains vector generalizations
of the Yajima-Oikawa and Melnikov \cite{Mf,M4} hierarchies,
respectively.

In \cite{CY,KSO,OC}, a k-constrained modified KP (k-cmKP) hierarchy
was introduced and investigated. Its Lax representation has the form
\begin{equation}
[\tilde{L}_k,\tilde{M}_n]=0,\,\,\tilde{L}_k=\tilde{B}_k+{\bf q}{\cal
M}_0D^{-1}{\bf r}^{\top}D,\,\,
\tilde{M}_n=\partial_{t_n}-\tilde{B}_n.
\end{equation}
where $\tilde{B}_k=D^k+\sum_{j=1}^{k-1}w_{j} D^j$.
For $k=1,2,3$, this leads to vector generalizations of the
Chen-Lee-Liu, the modified multi-component Yajima-Oikawa and
Melnikov hierarchies.

An essential extension of the k-cKP hierarchy is its
(2+1)-dimensional generalization \cite{MSS,6SSS}, given by
\begin{equation}\label{spa1}
\begin{array}{l}
 L_k=\beta_k\partial_{\tau_k}-B_k-{\bf q}{\cal M}_0D^{-1}{\bf r}^{\top}, \quad
 M_n=\alpha_n\partial_{t_n}-{A}_n, \\
 B_k=D^k+\sum_{j=0}^{k-2}u_jD^j, \quad
 {A}_n=D^n+\sum_{i=0}^{n-2}v_iD^i, \quad \\
 u_j=u_j(x,\tau_k,t_n),\,\,v_i=v_i(x,\tau_k,t_n), \quad
\alpha_n,\beta_k\in{\mathbb{C}},
\end{array}
\end{equation}
where $u_j$ and $v_i$ are scalar functions, ${\bf q}$ and ${\bf r}$
are $m$-component vector-functions. Lax equation $[L_k,M_n]=0$ is
equivalent to the system:
\begin{equation}\label{LMcom}
[L_k,M_n]_{\geq0}=0,M_n\{{\bf q}\}=0,\,\,M^{\tau}_n\{{\bf r}\}=0.
\end{equation}
System (\ref{LMcom}) can be rewritten as
\begin{equation}\nonumber
\alpha_{n} { {B_{k,t_n}}}=\beta_kA_{n,\tau_k}+[A_{n}, B_{k}] +
\left( [A_{n}, {\bf q} {\cal M}_0{D}^{-1} {\bf r}^{\top}]
\right)_{\geq0}, M_n\{{\bf q}\}=0,\,\,M^{\tau}_n\{{\bf r}\}=0.
\end{equation}
We list some members of this (2+1)-dimensional generalization of the
k-cKP hierarchy:
\renewcommand{\theenumi}{\arabic{enumi}}
\begin{enumerate}
\item $k=0,$ $n=2$.
\begin{equation}\label{L0M2}
{{L}_{0}}=\beta_0\partial_{\tau_0}-{\bf{q}}{{\mathcal{M}}_{0}}{{D}^{-1}}{{\bf{r}}^{\top
}},\,\,{{M}}_{2}=\alpha_2 {{\partial }_{{{t}_{2}}}}-{{D}^{2}}-v_0.
\end{equation}
The commutator equation $[L_0,M_{2}]=0$ is equivalent to the system:
\begin{equation}\label{DSnrnew}
\begin{array}{c}
 \alpha_2{\bf q}_{t_2}={\bf q}_{xx}+v_0{\bf q},\,\,
 -\alpha_2{\bf r}^{\top}_{t_2}={\bf r}^{\top}_{xx}+{\bf
 r}^{\top}v_0,\,\,
 \beta_0 v_{0,\tau_0}=-2({\bf q}{\cal M}_0{\bf r}^{\top})_x.
\end{array}
\end{equation}
 After the reduction $\beta_0\in{\mathbb{R}}$, $\alpha_2\in i{\mathbb{R}}$,
${\bf r}=\bar{\bf q}$, ${\cal M}_0={\cal M}^*_0$, $v_0=\bar{v}_0$,
the operators $L_1$ and $M_2$ in (\ref{L0M2}) are skew-Hermitian and
Hermitian, respectively, and (\ref{DSnrnew}) becomes the DS-III
system:
\begin{equation}\label{DSnrnew2x}
\begin{array}{c}
 \alpha_2{\bf q}_{t_2}={\bf q}_{xx}+v_0{\bf q},\,\,
 \beta_0 v_{0,\tau_0}=-2({\bf q}{\cal M}_0{\bf q}^{*})_x.
\end{array}
\end{equation}
\item
 $k=1$, $n=2$. Then (\ref{spa1}) has the form
\begin{equation}\label{L}
  L_1=\beta_{1} \partial_{\tau_{1}} - {D} -{\bf q}{\cal M}_0 {D}^{-1} {\bf r}^{\top}, \quad
 M_2=\alpha_{2} \partial_{t_{2}} - {D}^{2} - v_0,
\end{equation}
and the equation $\left[ L_1, M_2\right] =0,$ is equivalent to the
system,
\begin{equation}\label{Dsw}
\begin{array}{lcl}
 \alpha_{2} {\bf q}_{t_2}={\bf q}_{xx} + v_0{\bf q}, \quad
 \alpha_{2}  {\bf r}_{t_2}=-{\bf r}_{xx} - v_0{\bf r}, \\
\beta_{1} {v}_{0,\tau_1}={v}_{0,x} - 2 ({\bf q}{\cal M}_0{\bf
r}^{\top})_{x}.
\end{array}
\end{equation}
 After the reduction $\beta_1\in{\mathbb{R}}$, $\alpha_2\in i{\mathbb{R}}$,
${\bf r}=\bar{\bf q}$, ${\cal M}_0={\cal M}^*_0$, $v_0=\bar{v}_0$,
we obtain the DS-III system in the following form:
\begin{equation}\label{DeS}
\begin{array}{lcl}
\alpha_{2} {\bf q}_{t_2}={\bf q}_{xx} + v_0{\bf q},\quad \beta_{1}
v_{0,\tau_1}=v_{0,x} - 2 ({\bf q}{\cal M}_0{\bf q}^*)_{x}.
\end{array}
\end{equation}
In Remark \ref{remark1} we will show that DS-III systems
(\ref{DSnrnew2x}) and (\ref{DeS}) as well as all the other equations
of the hierarchy (\ref{spa1}) related to the operators $L_0$ and
$L_1$ in (\ref{spa1}) are equivalent via linear change of
independent variables.
\item
$k=1$, $n=3$. Now (\ref{spa1}) becomes
\begin{equation}\label{L1M3ea}
L_1=\beta_{1} \partial_{\tau_{1}} - {D} - {\bf q}{\cal M}_0 {D}^{-1}
{\bf r}^{\top}, \quad M_3=\alpha_{3}\partial_{t_{3}} -{D}^3- v_{1}
{D} - v_{0}.
\end{equation}
After the additional reduction $\alpha_3,\beta_1\in{\mathbb{R}}$,
${\cal M}_0={\cal M}^*_0$, $v_1=\bar{v}_1$, $\bar{v}_0+v_0=v_{1x}$,
the operators $L_1$, $M_3$ in (\ref{L1M3ea}) are skew-Hermitian, and
the Lax equation $[L_1,M_3]=0$ is equivalent to the following
(2+1)-dimensional generalization of the mKdV system:
\begin{equation}\label{mKDV}
\begin{array}{c}
\alpha_{3}{\bf q}_{t_3}  =  {\bf q}_{xxx} + v_{1}{\bf q}_{x} +
v_{0}{\bf q},\\ \beta_{1}  {v_{0,\tau_1}}
 =  {v_{0,x}} - 3 ({\bf q}_{x}{\cal M}_0{\bf q}^*)_{x}, \,\,
\beta_{1}  {v_{1,\tau_1}} = {v_{1,x}} - 3 ({\bf q}{\cal M}_0{\bf
q}^*)_{x}.
\end{array}
\end{equation}
This system admits the real version (${\cal M}^{\top}_0={\cal
M}^*_0$, ${\bf q}^*={\bf q}^{\top}$)
\begin{equation}\label{cmKDV}
\begin{array}{c}
\alpha_{3}{\bf q}_{t_3}  =  {\bf q}_{xxx} + v_{1}{\bf q}_{x} +
\frac12 v_{1,x}{\bf q},\,\,\,\, \beta_{1}  {v_{1,\tau_1}} =
{v_{1,x}} - 3 ({\bf q}{\cal M}_0{\bf q}^{\top})_{x}.
\end{array}
\end{equation}

\item
$k=2$, $n=2$. (\ref{spa1}) takes the form
\begin{equation}\label{L2M2y}
 L_2=\beta_{2} \partial_{\tau_{2}} - {D}^{2} - u_{0}
   - {\bf q}{\cal M}_0 {D}^{-1} {\bf r}^{\top}, \quad
 M_2=\alpha_{2} \partial_{t_{2}} - {D}^{2} - u_{0}.
\end{equation}
Under the reduction $\alpha_2,\beta_2\in i{\mathbb{R}}$, ${\bf
r}=\bar{\bf q}$, ${\cal M}_0=-{\cal M}_0^*$, $u_0=\bar{u}_0:=2u$,
the operators $L_2$ and $M_2$ in (\ref{L2M2y}) become Hermitian and
the Lax equation $[L_2,M_2]=0$ is equivalent to the following
system:
\begin{equation}\label{Yajima-Oikawa2+1}
\begin{array}{c}
\alpha_{2} {\bf q}_{t_2} =  {\bf q}_{xx} + 2u{\bf q},\qquad
\alpha_{2} u_{t_2} = \beta_{2}u_{\tau_{2}} + ({\bf q}{\cal M}_0{\bf
q}^*)_{x},
\end{array}
\end{equation}
which is a (2+1)-dimensional vector generalization of the
Yajima-Oikawa system.

\item $k=2$, $n=3$. Now (\ref{spa1}) becomes
\begin{equation}\label{398}
\begin{array}{c}
L_2=\beta_2\partial_{\tau_2}- D^2-2u_0-{\bf q}{\cal M}_0{D}^{-1}{\bf
r}^{\top}, \,\,
\\
M_3=\alpha_3\partial_{t_3}-D^3-3u_0D-\frac32
\left(u_{0,x}+\beta_2D^{-1}\{u_{0,\tau_2}\}+{\bf q}{\cal M}_0{\bf
r}^{\top}\right). \,\,\,
\end{array}
\end{equation}
In formula (\ref{398}), $D^{-1}\{u_{0,\tau_2}\}$ denotes indefinite
integral of the function $u_{0,\tau_2}$ with respect to $x$. With
the additional reduction $\beta_2\in i{\mathbb{R}}$,
$\alpha_3\in{\mathbb{R}}$, $u_0=\bar{u}_0:=u$, ${\cal M}_0=-{\cal
M}^*_0$ and ${\bf r}=\bar{\bf q}$, this is equivalent to the
following generalization of the higher Yajima-Oikawa system
\cite{PHD} (in scalar case ($m=1$) Manakov LAB-triad for this
equation was proposed in \cite{Mf}):
\begin{equation}\label{3100}
\begin{array}{l}
 \alpha_3{\bf q}_{t_3}\!=\!  {\bf q}_{xxx}\! + 3 u{\bf q}_x
 +\frac32\left(u_x+\beta_2D^{-1}\{u_{\tau_2}\}+{\bf q}{\cal M}_0{\bf q}^*\right){\bf q},\\
 \left[\alpha_3u_{t_3}-\frac14u_{xxx}-3uu_x
  +\frac34\right.  \left(\!{\bf q}{\cal M}_0{\bf q}_x^*-{\bf q}_x{\cal M}_0{\bf q}^*\!\right)_x
 \left.
 - \frac34 \beta_2\left({\bf q}{\cal M}_0{\bf
 q}^*\right)_{\tau_2}\right]_{x}\!
   =\!\frac34\beta_2^2\! u_{{\tau}_2{\tau}_2}.
\end{array}
\end{equation}
\item $k=3,n=3$. (\ref{spa1}) takes the form
\begin{equation}
L_3=\beta_3\partial_{\tau_3}-D^3-u_1D-\frac12 u_{1,x}-{\bf q}{\cal
M}_0D^{-1}{\bf r}^{\top},\quad
M_3=\alpha_3\partial_{t_3}-D^3-u_1D-\frac12u_{1,x}.
\end{equation}
With the additional reduction $\alpha_3,\beta_3\in{\mathbb{R}}$,
$u_1=\bar{u}_1:=u$, ${\cal M}_0={\cal M}^*_0=\bar{{\cal M}}_0$ and
${\bf r}=\bar{\bf q}={\bf q}$, this is equivalent to the following
(2+1)-dimensional generalization of the Drinfeld-Sokolov-Wilson
equation \cite{DS,WG,HGR}:
\begin{equation}\label{DSW}
\alpha_3{\bf q}_{t_3}={\bf q}_{xxx}+u{\bf q}_x+\frac12u_x{\bf q},\,
\alpha_3u_{t_3}=\beta_3u_{\tau_3}+3({\bf q}{\cal M}_0{\bf
q}^{\top})_x.
\end{equation}
\end{enumerate}
 The
following remark establishes connections between equations
represented by Lax pairs ($L_0$, $M_k$) and ($L_1$, $M_k$), ($L_k$,
$M_k$) and ($L_1$, $M_k$) respectively (see formulae (\ref{spa1})).
\begin{remark}\label{remark1}
The change of variables $\tilde{x}:=x-\frac{1}{\beta_0}\tau_0$,
$\tau_1 :=\frac{\beta_1}{\beta}_0\tau_0$, $\tilde{\bf
q}(\tilde{x},\tau_1):={\bf q}(x,\tau_0)$,
$\tilde{v}_0(\tilde{x},\tau_1): =v_0(x,\tau_0)$  maps DS-III
equation (\ref{DSnrnew2x})
to another form (\ref{DeS}) of it. The same change maps the "higher"
equations related to operator $L_0$ to the equations related to
$L_1$ (see formulae (\ref{spa1})).
\end{remark}
\begin{remark}
The change of variables
\begin{equation}
\begin{array}{l}
\tilde{x}:=x-\frac{1}{\beta_k}\tau_k,
{\tilde{t}}_k:=t_k+\frac{\alpha_k}{\beta_k}\tau_k,
\tau_1:=\frac{\beta_1}{\beta_k}\tau_k,
v_{k-2}({\tilde{x}},{\tilde{t}}_k,\tau_1):=u_{k-2}(x,t_k,\tau_k),\\
\tilde{{\bf q}}({\tilde{x}},{\tilde{t}}_k,\tau_1):={\bf
q}(x,t_k,\tau_k),
\end{array}
\end{equation}
 maps (2+1)-dimensional generalization of Yajima-Oikawa system (\ref{Yajima-Oikawa2+1}) to DS-III system (\ref{DeS})
  in the case $k=2$. In the case $k=3$
 it maps (2+1)-extension of Drinfeld-Sokolov-Wilson equation
 (\ref{DSW}) to the real version of (2+1)-extension of mKdV equation
 (\ref{cmKDV}). An analogous relation holds between all other equations
 with Lax pairs ($L_k$, $M_k$) and ($L_1$, $M_k$) where the operators $L_1$, $L_k$
 and $M_k$ are defined by formulae (\ref{spa1}).
\end{remark}

Thus, for $k=1$ we have the DS-III hierarchy (its first members are
DS-III ($k=1,n=2$) and a special (2+1)-dimensional extension of mKdV
($k=1,n=3$), see (\ref{DeS}) and (\ref{mKDV}). For $k=2$, $k=3$, we
have (2+1)-dimensional generalizations of the Yajima-Oikawa (in
particular, it contains (\ref{Yajima-Oikawa2+1}) and (\ref{3100}))
and the Melnikov hierarchy \cite{M4}, respectively.

\section{Bidirectional generalizations of (2+1)-dimensional k-constrained KP
hierarchy}\label{extended} In this section we introduce a new
generalization of  the (2+1)-dimensional k-constrained KP hierarchy
given by (\ref{spa1}) to the case of two integro-differential
operators. One of them (the operator $M_{n,l}$ in (\ref{ex2+1}))
generalizes the corresponding operator $M_n$ (\ref{spa1}) and
depends on two independent indices $l$ and $n$. It leads to
generalization of (2+1)-dimensional k-cKP hierarchy (\ref{spa1}) in
additional direction $l$ ($l=1,2,\ldots$). We do not consider the
case $l=0$ because of the Remark \ref{REM}. For further purposes we
will use the following well-known formulae for integral operator
$h_1D^{-1}h_2$ constructed by matrix-valued functions $h_1$ and
$h_2$ and the differential operator $A$ with matrix-valued
coefficients in the algebra of pseudodifferential
operators: 
\begin{equation}\label{Sydorenko:eq25}
Ah_1{\cal D}^{-1} h_2=(Ah_1{\cal D}^{-1} h_2)_{\geq0}+ A\{h_1\}{\cal
D}^{-1} h_2,
\end{equation}
\begin{equation}\label{Sydorenko:eq26}
h_1{\cal D}^{-1} h_2 A=(h_1{\cal D}^{-1} h_2 A)_{\geq0}+ h_1{\cal
D}^{-1} [A^\tau \{h_2^{\top}\}]^\top,
\end{equation}
\begin{equation}\label{Sydorenko:eq27}
h_1{\cal D}^{-1} h_2 h_3{\cal D}^{-1} h_4=h_1D^{-1}\{h_2h_3\}{\cal
D}^{-1} h_4 -h_1{\cal D}^{-1}D^{-1}\{h_2h_3\} h_4.
\end{equation}

Consider the following bidirectional generalization of
(2+1)-dimensional k-constrained KP hierarchy (\ref{spa1}):
\begin{equation}\label{ex2+1}
\begin{array}{l}
L_k=\beta_k\partial_{\tau_k}-B_k-{\bf q}{\cal M}_0D^{-1}{\bf
r}^{\top},\,\, B_k=\sum_{j=0}^{k}u_jD^j,u_j=u_j(x,\tau_k,t_n),\,
\beta_k\in{\mathbb{C}},\\
M_{n,l}=\alpha_n\partial_{t_n}-{A}_n-c_l\sum_{j=0}^l{\bf q}[j]{\cal
M}_0D^{-1}{\bf r}^{\top}[l-j],\,\,\, l=1,\ldots
\\{A}_n=\sum_{i=0}^{n}v_iD^i,
v_i=v_i(x,\tau_k,t_n), \alpha_n\in{\mathbb{C}},
\end{array}
\end{equation}
where $u_j$ and $v_i$, ${\bf q}$ and ${\bf r}$  are $N\times N$ and
$N\times M$ matrix functions respectively; ${\bf q}[j]$ and ${\bf
r}[j]$ are matrix functions of the following form: $
{\bf{q}}[j]:=(L_k)^j\{{\bf q}\},\,\, {\bf
r}^{\top}[j]:=((L_k^{\tau})^j\{{\bf r}\})^{\top}. $

The following theorem  holds.
\begin{theorem}\label{T1}
The Lax equation $[L_k,M_{n,l}]=0$ is equivalent to the system:
 \begin{equation}\label{fre}
 [L_k,M_{n,l}]_{\geq0}=0,M_{n,l}\{{\bf q}\}=c_l(L_k)^{l+1}\{{\bf q}\},\,M_{n,l}^{\tau}\{{\bf{r}}\}=c_l(L_k^{\tau})^{l+1}\{{\bf{r}}\}.
 \end{equation}
 \end{theorem}
 \begin{proof}
 From the equality $[L_k,M_{n,l}]=[L_k,M_{n,l}]_{\geq0}+[L_k,M_{n,l}]_{<0}$ we obtain that the Lax equation $[L_k,M_{n,l}]=0$ is equivalent to the following one:
\begin{equation}
[L_k,M_{n,l}]_{\geq0}=0,\,\,[L_k,M_{n,l}]_{<0}=0.
\end{equation}
Thus, it is sufficient to prove that
$[L_k,M_{n,l}]_{<0}=0\Longleftrightarrow$
$(M_{n,l}-c_l(L_k)^{l+1})\{{\bf
q}\}=(M_{n,l}^{\tau}-c_l(L_k^{\tau})^{l+1})\{{\bf r}\}=0$. Using
bi-linearity of the commutator and explicit form (\ref{ex2+1}) of
the operators $L_k$ and $M_{n,l}$ we obtain:
\begin{equation}\label{Frt}
\begin{array}{c}
[L_k,M_{n,l}]_{<0}=c_l\sum_{j=0}^l[{\bf q}[j]{\cal M}_0D^{-1}{\bf
r}^{\top}[l-j],\beta_k\partial_{\tau_k}-B_k]_{<0}+\\+c_l\sum_{j=0}^l[{\bf
q}{\cal M}_0D^{-1}{\bf r}^{\top},{\bf q}[j]{\cal M}_0D^{-1}{\bf
r}^{\top}[l-j]]_{<0}+[\alpha_n\partial_{t_n}-A_n,{\bf q}{\cal
M}_0D^{-1}{\bf r}^{\top}]_{<0}.
\end{array}
\end{equation}

After direct computations of each of the three items on the
right-hand side of formula (\ref{Frt}) we obtain:

\begin{enumerate}
\item
\begin{equation}\label{equ1}
\begin{array}{c}
c_l\sum_{j=0}^l[{\bf q}[j]{\cal M}_0D^{-1}{\bf
r}^{\top}[l-j],\beta_k\partial_{\tau_k}-B_k]_{<0}=-c_l\sum_{j=0}^l\left(\beta_k{\bf
q}_{\tau_k}[j]-B_k\{{\bf q}[j]\}\right)\cdot\\\cdot{\cal
M}_0D^{-1}{\bf r}^{\top}[l-j]-c_l\sum_{j=0}^l{\bf q}[j]{\cal
M}_0D^{-1}\left(\beta_k{\bf r}^{\top}_{\tau_k}[l-j]+B_k^{\tau}\{{\bf
r}^{\top}[l-j]\}\right).
\end{array}
\end{equation}
Equality (\ref{equ1}) is a consequence of formulae
(\ref{Sydorenko:eq25})-(\ref{Sydorenko:eq26}).
\item

\begin{equation}\label{item2}
\begin{array}{c}
c_l\sum_{j=0}^l[{\bf q}{\cal M}_0D^{-1}{\bf r}^{\top},{\bf
q}[j]{\cal M}_0D^{-1}{\bf r}^{\top}[l-j]]_{<0}=c_l\sum_{j=0}^l {\bf
q}{\cal M}_0D^{-1}\{{\bf r}^{\top}{\bf q}[j]\}\times\\\times{\cal
M}_0D^{-1}{\bf r}^{\top}[l-j]-c_l\sum_{j=0}^l {\bf q}{\cal
M}_0D^{-1}D^{-1}\{{\bf r}^{\top}{\bf q}[j]\}{\cal M}_0{\bf
r}^{\top}[l-j]-\\- c_l\sum_{j=0}^l {\bf q}[j]{\cal M}_0D^{-1}\{{\bf
r}^{\top}[l-j] {\bf q}\}{\cal M}_0D^{-1}{\bf
r}^{\top}+\\+c_l\sum_{j=0}^l {\bf q}[j]{\cal M}_0D^{-1}D^{-1}\{{\bf
r}^{\top}[l-j] {\bf q}\}{\cal M}_0{\bf r}^{\top}.
\end{array}
\end{equation}
Formula (\ref{item2}) follows from (\ref{Sydorenko:eq27}).
\item

\begin{equation}\label{dsa}
\begin{array}{l}
[\alpha_n\partial_{t_n}-A_n,{\bf q}{\cal M}_0D^{-1}{\bf
r}^{\top}]_{<0}=\left(\alpha_n{\bf q}_{t_n}-A_n\{{\bf
q}\}\right){\cal M}_0D^{-1}{\bf r}^{\top}+\\+{\bf q}{\cal
M}_0D^{-1}\left(\alpha_n{\bf r}^{\top}_{t_n}+(A_n^{\tau}\{{\bf
r}\})^{\top}\right).
\end{array}
\end{equation}
The latter equality is obtained via
(\ref{Sydorenko:eq25})-(\ref{Sydorenko:eq26}).
\end{enumerate}
From formulae (\ref{Frt})-(\ref{dsa}) we have
\begin{equation}
\begin{array}{l}
[L_k,M_{n,l}]_{<0}=c_l\left(\!\sum_{j=0}^l{{\bf q}}[j]{\cal
M}_0D^{-1}L_k^{\tau}\{{\bf r}[l-j]\}-\sum_{j=0}^lL_k\{{\bf
q}[j]\}{\cal M}_0D^{-1}{\bf
r}^{\top}[l-j]\!\right)\!+\\+M_{n,l}\{{\bf q}\}{\cal M}_0D^{-1}{\bf
r}^{\top}-{\bf q}{\cal M}_0D^{-1} (M_{n,l}^{\tau}\{{\bf
r}\})^{\top}=c_l\sum_{j=0}^l{\bf q}[j]{\cal M}_0D^{-1}{\bf
r}^{\top}[l-j+1]-\\-c_l\sum_{j=0}^l{\bf q}[j+1]{\cal M}_0D^{-1}{\bf
r}^{\top}[l-j]+M_{n,l}\{{\bf q}\}{\cal M}_0D^{-1}{\bf r}^{\top}-{\bf
q}{\cal M}_0D^{-1}(M^{\tau}_{n,l}\{{\bf
r}\})^{\top}=\\=(M_{n,l}-c_l(L_k)^{l+1})\{{\bf q}\}{\cal
M}_0D^{-1}{\bf r}^{\top}-{\bf q}{\cal
M}_0D^{-1}((M_{n,l}^{\tau}-c_l(L_k^{\tau})^{l+1})\{{\bf
r}\})^{\top}.
\end{array}
\end{equation}
From the last equality we obtain the equivalence of the equation
$[L_k,M_{n,l}]=0$ and (\ref{fre}).
\end{proof}
New hierarchy (\ref{ex2+1}) consists of several special cases:
\begin{enumerate}
\item $\beta_k=0$, $c_l=0$. Under this assumption we obtain Matrix
k-constrained KP hierarchy \cite{Oevel93}. We shall also point out
that the case $\beta_k=0$ and $c_l\neq0$ also reduces to Matrix
k-constrained KP hierarchy.
\item $\alpha_n=0$. In this case we obtain a new (1+1)-dimensional Matrix
k-constrained KP hierarchy \cite{NEW}. Their members are stationary
with respect to $t_n$  members of (2+1)-BDk-cKP hierarchy
(\ref{ex2+1}).
\item $c_l=0$, $N=1$, $v_n=u_k=1$, $v_{n-1}=u_{k-1}=0$. In this case we obtain (2+1)-dimensional k-cKP hierarchy
given by (\ref{spa1}).
\item $n=0$. In this case the differential part of ${M}_{0,l}$ in
(\ref{ex2+1}) is equal to zero: $A_0=0$. We obtain a new
generalization of DS-III hierarchy.
\item $c_l=0$. We obtain $(t_A,\tau_B)$-Matrix KP Hierarchy
that was investigated in \cite{Zeng}.
\item If we set $l=0$ we obtain $(\gamma_A,\sigma_B)$-Matrix KP
hierarchy that was investigated in \cite{YYQB}. However, we do not
include $l=0$ in formula (\ref{ex2+1}) because of the following
remark.
\begin{remark}\label{REM}
In the case $l=0$ we obtain:
\begin{equation}\label{AS}
M_{n,0}=\alpha_n\partial_{t_n}-{A}_n-c_0{\bf q}{\cal M}_0D^{-1}{\bf
r}^{\top}=(\alpha_n\partial_{t_n}-c_0\beta_k\partial_{\tau_k})-({A}_n-c_0B_k)+c_0L_k.
\end{equation}
The last summand in the latter formula can be ignored because $L_k$
commutes with itself. Thus, we obtained an evolution differential
operator again. It means that in the case $l=0$ the hierarchy given
by the operators $M_{n,0}$ and $L_k$ after the change of independent
variables
($\alpha_n\partial_{t_n}-c_0\beta_k\partial_{\tau_k}\rightarrow
\tilde{\alpha}_n\partial_{\tilde{t}_n}$) coincides with hierarchy
(\ref{spa1}) in the matrix case.

\end{remark}

\end{enumerate}

 The following corollary follows from Theorem \ref{T1}:
\begin{corollary}
The Lax equation $[L_k,\tilde{M}_{n,l}]=0$, where
\begin{equation}\label{tM}
\tilde{M}_{n,l}=M_{n,l}-c_l(L_k)^{l+1},
\end{equation}
 and the operators $L_k$ and $M_n$ are defined by (\ref{ex2+1}), is equivalent to the system:
 \begin{equation}\label{frex}
 [L_k,\tilde{M}_{n,l}]_{\geq0}=0,\tilde{M}_{n,l}\{{\bf q}\}=0,\,\tilde{M}_{n,l}^{\tau}\{{\bf{r}}\}=0.
 \end{equation}
\end{corollary}
\begin{proof}
It is evident that the operator $L_k$ commutes with the natural
powers of itself: $[L_k,(L_k)^{l+1}]=0$, $l\in{\mathbb{N}}$. Thus by
bilinearity of the commutator  we obtain that
$[L_k,\tilde{M}_{n,l}]=0$ holds if and only if $[L_k,M_{n,l}]=0$. It
remains to use Theorem \ref{T1} in order to complete the proof.
\end{proof}
As a result of the hierarchy given by operators (\ref{ex2+1}) and a
bilinearity of the commutator we obtain the following essential
generalization of (\ref{ex2+1}):
\begin{equation}\label{ex2+2}
\begin{array}{l}
L_k=\beta_k\partial_{\tau_k}-B_k-{\bf q}{\cal M}_0D^{-1}{\bf
r}^{\top},\,\, B_k=\sum_{j=0}^{k}u_jD^j,u_j=u_j(x,\tau_k,t_n),\,
\beta_k\in{\mathbb{C}},\\
\!P_{n,m}\!=\!\alpha_n\partial_{t_n}\!-\!{A}_n-\sum_{l=1}^mc_l\left(\left(\sum_{j=0}^l{\bf
q}[j]{\cal
M}_0D^{-1}{\bf r}^{\top}[l-j]\right)+(L_k)^{l+1}\right),\,\,\, m=1,2,\ldots
\\{A}_n=\sum_{i=0}^{n}v_iD^i,
v_i=v_i(x,\tau_k,t_n), \alpha_n\in{\mathbb{C}},
\end{array}
\end{equation}
and the following corollary that is immediate consequence of Theorem
\ref{T1}:
\begin{corollary}
The commutator equation $[L_k,P_{n,m}]=0$ is equivalent to the
following system:
\begin{equation}\label{frexxx}
 [L_k,P_{n,m}]_{\geq0}=0,P_{n,m}\{{\bf q}\}=0,\,P_{n,m}^{\tau}\{{\bf{r}}\}=0.
 \end{equation}
\end{corollary}

 For further convenience we will consider the Lax pairs
consisting of the operators $L_k$ (\ref{ex2+1}) and
$\tilde{M}_{n,l}$ (\ref{tM}) (the operator $\tilde{M}_{n,l}$ is
involved in equations for functions ${\bf q}$ and ${\bf r}$; see
formulae (\ref{frex})).
 Consider examples of equations given by operators $L_k$
 (\ref{ex2+1}) and $\tilde{M}_{n,l}$ (\ref{tM})
that can be obtained under certain choice of $(k,n,l)$.



 I. $k=0.$

In this case the operator $L_0$ (\ref{ex2+1}) has the form:
\begin{equation}
{{L}_{0}}=\beta_0\partial_{\tau_0}-{\bf{q}}{{\mathcal{M}}_{0}}{{D}^{-1}}{{\bf{r}}^{\top
}}.
\end{equation}
For further simplicity we use the change of variables $\tau_0:=y$,
$\beta:=\beta_0$:
\begin{equation}\label{L0}
{{L}_{0}}=\beta\partial_{y}-{\bf{q}}{{\mathcal{M}}_{0}}{{D}^{-1}}{{\bf{r}}^{\top
}}.
\end{equation}
\begin{enumerate}
%
\item $n=2$, $l=1$.
\begin{equation}\label{1eq}
\begin{array}{l}
\tilde{{M }}_{2,1}=M_{2,1}-c_1(L_0)^2=\alpha_2 {{\partial
}_{{{t}_{2}}}}-c({{D}^{2}}+v_0)-c_1\beta^2\partial^2_y+2c_1\beta{\bf
q}{\cal M}_0D^{-1}{\bf r}^{\top}_y+\\+2c_1\beta{\bf q}{\cal
M}_0D^{-1}{\bf r}^{\top}\partial_y,\,\,c\in{\mathbb{C}}.
\end{array}
\end{equation}
If $c=0$ we obtain that $\tilde{M}_{2,1}=\tilde{M}_{0,1}$. We use
$c$ in formula (\ref{1eq}) in order not to consider separately the
case $n=0$. The commutator equation $[L_0,\tilde{M}_{2,1}]=0$ is
equivalent to the system:
\begin{equation}\label{DSnr}
\begin{array}{c}
 \alpha_2{\bf q}_{t_2}=c{\bf q}_{xx}+c_1\beta^2{\bf q}_{yy}+cv_0{\bf q}+c_1{\bf
 q}{\cal M}_0S,\\
 -\alpha_2{\bf r}^{\top}_{t_2}=c{\bf r}^{\top}_{xx}+c_1\beta^2{\bf r}^{\top}_{yy}+c{\bf
 r}^{\top}v_0+c_1S{\cal M}_0{\bf r}^{\top},\\
 \beta v_{0y}=-2({\bf q}{\cal M}_0{\bf r}^{\top})_x,\,\,S_{x}=-2\beta({\bf r}^{\top}{\bf q})_y.
\end{array}
\end{equation}
The Lax pair $L_0$ (\ref{L0}) and $\tilde{M}_{2,1}$ (\ref{1eq}) and
the corresponding system (\ref{DSnr}) were investigated in \cite{n}.
Consider additional reductions of pair of the operators $L_0$
(\ref{L0}) and $\tilde{M}_{2,1}$ (\ref{1eq}) and system
(\ref{DSnr}). After the reduction $c,c_1\in\mathbb{R}$,
$\beta\in{\mathbb{R}}$, $\alpha_2\in i{\mathbb{R}}$; ${\bf
r}^{\top}={\bf q}^*$, ${\cal M}_0={\cal M}_0^*$, the operators $L_0$
and $\tilde{M}_{2,1}$ are skew-Hermitian and Hermitian respectively,
and (\ref{DSnr}) takes the form
\begin{equation}\label{DS}
\begin{array}{c}
 \alpha_2{\bf q}_{t_2}=c{\bf q}_{xx}+c_1\beta^2{\bf q}_{yy}+cv_0{\bf q}+c_1{\bf
 q}{\cal M}_0S,\\
 \beta v_{0y}=-2({\bf q}{\cal M}_0{\bf q}^*)_x,\,\,S_{x}=-2\beta({\bf q}^*{\bf q})_y.
\end{array}
\end{equation}

This has the following two interesting subcases:
\renewcommand{\theenumi}{\arabic{enumi}}
\begin{enumerate}
\item
$c_1=0$. Then we have
\begin{equation}\label{DSc2=0}
\begin{array}{c}
 \alpha_2{\bf q}_{t_2}=c{\bf q}_{xx}+cv_0{\bf q},\,\,\,
 \beta v_{0y}=-2({\bf q}{\cal M}_0{\bf q}^*)_x.
\end{array}
\end{equation}

\item $c=0$. Then (\ref{DS}) takes the form
\begin{equation}\label{DSc1=0}
\begin{array}{c}
 \alpha_2{\bf q}_{t_2}=c_1\beta^2{\bf q}_{yy}+c_1{\bf
 q}{\cal M}_0S,\,\,\,\
S_{x}=-2\beta({\bf q}^*{\bf q})_y.
\end{array}
\end{equation}
\end{enumerate}
Systems (\ref{DSc2=0}) and (\ref{DSc1=0}) are two different matrix
generalizations of the Davey-Stewartson equation (DS-III-a and
DS-III-b). In the scalar case equation (\ref{DSc2=0}) was
investigated in \cite{Zakharov,fokas}. The vector version of DS-III
(\ref{DSc2=0}) and its Lax representation given by operators
(\ref{L0}) and (\ref{1eq}) in case $c_1=0$, $N=1$ were introduced in
\cite{OC,MSS,6SSS}. 

Let us consider (\ref{DS}) in the case where $\beta=1$, $u:={\bf q}$
and ${\cal M}_0:=\mu$ are scalars. Then (\ref{DS}) becomes
\begin{equation}\label{DSuscal}
\begin{array}{c}
 \alpha_2u_{t_2}=cu_{xx}+c_1u_{yy}+vu+\mu c_1 S u,\,\,
 v_{0y}=-2\mu(|u|^2)_x,\,\,S_{x}=-2(|u|^2)_y.
\end{array}
\end{equation}
Setting $c=c_1=1$ and $\mu=1$, as a consequence of (\ref{DSuscal})
we obtain
\begin{equation}\label{DS2}
\begin{array}{c}
 \alpha_2u_{t_2}=u_{xx}+u_{yy}+S_{1}u,\,\,\,
 S_{1,xy}=-2(|u|^2)_{xx}-2(|u|^2)_{yy},
\end{array}
\end{equation}
where $S_1=v_0+S$. This is the well-known Davey-Stewartson system
(DS-I) and (\ref{DS}) is therefore a matrix (noncommutative)
generalization. The interest in noncommutative versions of DS
systems (in particular, solution generating technique) has also
arisen recently in \cite{DMH,GM}. In Section \ref{dressed} we will
consider dressing method for matrix generalization of DS system
(\ref{DS}) that leads to its exact solutions.
  In \cite{n} DS-II system was also
obtained from system (\ref{DSnr}) and its Lax pair (\ref{L0}),
(\ref{1eq}) after the change $x\rightarrow z$, $y\rightarrow
\bar{z}$.

\item $n=2$, $l=2$.
\begin{equation}\label{1eqka}
\begin{array}{l}
\tilde{{M }}_{2,2}=M_{2,2}-c_2(L_0)^3=\alpha_2 {{\partial
}_{{{t}_{2}}}}-c_2\beta^3\partial_y^3-{{D}^{2}}-v_0+3\beta^2c_2{\bf q}_y{\cal M}_0D^{-1}{\bf r}^{\top}_y  \\
      -3\beta c_2{\bf q}{\cal M}_0\partial_yD^{-1}{\bf r}^{\top}{\bf q}{\cal M}_0D^{-1}{\bf r}^{\top}
  + 3\beta c_2{\bf q}{\cal M}_0D^{-1}\{{\bf r}^{\top}{\bf q}\}_y{\cal M}_0D^{-1}{\bf r}^{\top} \\
  +3c_2\beta^2\partial_y{\bf q}{\cal M}_0D^{-1}{\bf
  r}^{\top}\partial_y.
  \end{array}
\end{equation}
The commutator equation $[L_0,\tilde{M}_{2,2}]=0$ is equivalent to
the system:
\begin{equation}\label{aDSL1M3}
\begin{array}{l}
\alpha_2{\bf q}_{t_2}-{\bf q}_{xx}-c_2\beta^3{\bf q}_{yyy}- v_0{\bf
q}+ 3c_2\beta{\bf q}_y{\cal M}_0S_1+ 3\beta c_2{\bf q}{\cal
M}_0S_{1y}
-3c_2{\bf q}{\cal M}_0S_2 \\
 -3c_2{\bf q}D^{-1}\left\{{\cal M}_0{\bf r}^{\top}{\bf q}{\cal
M}_0S_1-{\cal M}_0S_1{\cal M}_0{\bf r}^{\top}{\bf q}\right\}=0, \\
 -\alpha_2{\bf r}^{\top}_{t_2}-{\bf r}^{\top}_{xx}-\beta^3c_2{\bf
r}^{\top}_{yyy}-{\bf r}^{\top}v_0+ 3\beta c_2S_1{\cal M}_0{\bf
r}^{\top}_y+3c_2S_2{\cal M}_0{\bf r}^{\top}  \\
 - 3c_2D^{-1}\left\{S_1{\cal M}_0{\bf r}^{\top}{\bf q}-{\bf
r}^{\top}{\bf q}{\cal
M}_0S_1\right\}{\cal M}_0{\bf r}^{\top}=0, \\
\beta v_{0y}=-2({\bf q}{\cal M}_0{\bf r}^{\top})_x,\,\,
S_{1x}=\beta({\bf{r}}^{\top}{\bf q})_y,\,\,\, \,S_{2x}=\beta^2({\bf
r}^{\top}_y{\bf q})_y.
\end{array}
\end{equation}




\item $n=3,l=1$.

In this case the operator $\tilde{M}_{3,1}$ has the form:
\begin{equation}\label{L1M3e}
\begin{array}{l}
\tilde{M}_{3,1}=M_{3,1}-c_1(L_0)^2=\alpha_{3}\partial_{t_{3}}
-{D}^3- v_{1} {D} - v_{0}-c_1\beta^2\partial^2_y+\\+2c_1\beta{\bf
q}{\cal M}_0D^{-1}{\bf r}^{\top}_y+2c_1\beta{\bf q}{\cal
M}_0D^{-1}{\bf r}^{\top}\partial_y.
\end{array}
\end{equation}
The equation $[L_0,\tilde{M}_{3,1}]=0$ is equivalent to the system:
\begin{equation}\label{mKDV2}
\begin{array}{l}
\alpha_{3}{\bf q}_{t_3}  =  {\bf q}_{xxx}+c_1\beta^2{\bf q}_{yy}+
v_{1}{\bf q}_{x} + v_{0}{\bf q}+c_1{\bf
 q}{\cal M}_0S_1,\\
-\alpha_{3}{\bf r}^{\top}_{t_3}  =-  {\bf
r}^{\top}_{xxx}+c_1\beta^2{\bf r}^{\top}_{yy} -({\bf
r}^{\top}v_{1})_x + {\bf r}^{\top}v_{0}+c_1S_1{\cal M}_0{\bf r}^{\top},\\
\beta{v_{0,y}}
 = - 3({\bf q}_{x}{\cal M}_0{\bf r}^{\top})_{x}+[{\bf q}{\cal M}_0{\bf
r}^{\top},v_1], \,\,\\ \beta  {v_{1,y}} =-3 ({\bf q}{\cal M}_0{\bf
r}^{\top})_{x},\,\,S_{1x}=-2\beta({\bf r}^{\top}{\bf q})_y.
\end{array}
\end{equation}
\item $n=3,l=2$

\begin{equation}\label{L1M3}
\begin{array}{l}
 \tilde{M}_{3,2}=\alpha_3\partial_{t_3}-cD^3-c_2\beta^3\partial_y^3-cv_1D-
cv_0+3\beta^2c_2{\bf q}_y{\cal M}_0D^{-1}{\bf r}^{\top}_y-  \\
      -3\beta c_2{\bf q}{\cal M}_0\partial_yD^{-1}{\bf r}^{\top}{\bf q}{\cal M}_0D^{-1}{\bf r}^{\top}
  + 3\beta c_2{\bf q}{\cal M}_0D^{-1}\{{\bf r}^{\top}{\bf q}\}_y{\cal M}_0D^{-1}{\bf r}^{\top} \\
  +3c_2\beta^2\partial_y{\bf q}{\cal M}_0D^{-1}{\bf
  r}^{\top}\partial_y.
\end{array}
\end{equation} The
Lax equation $[L_0,\tilde{M}_{3,2}]=0$ results in the system:
\begin{equation}\label{DSL1M3}
\begin{array}{l}
\alpha_3{\bf q}_{t_3}-c{\bf q}_{xxx}-c_2\beta^3{\bf q}_{yyy}-
cv_1{\bf q}_x+ 3c_2\beta{\bf q}_y{\cal M}_0S_1+ 3c_2\beta {\bf
q}{\cal M}_0S_{1y} -cv_0{\bf q}-\\-3c_2{\bf q}{\cal M}_0S_2-3c_2{\bf
q}D^{-1}\left\{{\cal M}_0{\bf r}^{\top}{\bf q}{\cal M}_0S_1
  -{\cal M}_0S_1{\cal M}_0{\bf r}^{\top}{\bf q}\right\}=0,\\
 \alpha_3{\bf r}^{\top}_{t_3}-c{\bf r}^{\top}_{xxx}-\beta^3c_2{\bf
r}^{\top}_{yyy}-c{\bf r}^{\top}_xv_1-c{\bf r}^{\top}v_{1x}+ 3\beta
c_2S_1{\cal M}_0{\bf r}^{\top}_y+ c{\bf r}^{\top}v_0
+\\+3c_2S_2{\cal M}_0{\bf r}^{\top} - 3c_2D^{-1}\left\{S_1{\cal
M}_0{\bf r}^{\top}{\bf q}-{\bf r}^{\top}{\bf q}{\cal
M}_0S_1\right\}{\cal M}_0{\bf r}^{\top}=0,\\
\beta v_{1,y}=-3({\bf q}{\cal M}_0{\bf r}^{\top})_x,\,
S_{1x}=\beta({\bf{r}}^{\top}{\bf q})_y,\, \\\beta v_{0y}=-3({\bf
q}_x{\cal M}_0{\bf r}^{\top})_x+[{\bf q}{\cal M}_0{\bf
r}^{\top},v_1],\,S_{2x}=\beta^2({\bf r}^{\top}_y{\bf q})_y.
\end{array}
\end{equation}
System (\ref{DSL1M3}) and its Lax pair $L_0$ (\ref{L0}) and
$\tilde{M}_{3,2}$ (\ref{L1M3}) was investigated in \cite{n}. We list
some reductions of this system:
\begin{enumerate}
\item $\alpha_3,\beta,c,c_2\in{\mathbb{R}}$, ${\bf r}^{\top}={\bf q}^*$, ${\cal M}_0={\cal M}^*_0$.
The operators $L_1$ and $M_3$ are then skew-Hermitian and
(\ref{DSL1M3}) takes the form
\begin{equation}\label{DSL1M3r}
\begin{array}{l}
 \alpha_3{\bf q}_{t_3}-c{\bf q}_{xxx}-c_2\beta^3{\bf q}_{yyy}-
cv_1{\bf q}_x+ 3c_2\beta{\bf q}_y{\cal M}_0S_1+ 3c_2\beta{\bf
q}{\cal M}_0S_{1y}- \\-cv_0{\bf q}-3c_2{\bf q}{\cal M}_0S_2
 - 3c_2{\bf q}D^{-1}\left\{{\cal M}_0{\bf q}^{*}{\bf q}{\cal
M}_0S_1-{\cal
M}_0S_1{\cal M}_0{\bf q}^{*}{\bf q}\right\}=0,  \\
 \beta v_{1y}=-3({\bf q}{\cal M}_0{\bf q}^*)_x,\, S_{1x}=\beta({\bf
q}^*{\bf{q}})_y,\,\\ \beta v_{0y}=-3({\bf q}_x{\cal M}_0{\bf
q}^*)_x+ [{\bf q}{\cal M}_0{\bf q}^*,v_1],\,S_{2x}=\beta^2({\bf
q}^*_y{\bf q})_y.
\end{array}
\end{equation}
In the scalar case ($N=m=1$), setting ${\mathbb{R}}\ni\mu:={\cal
M}_0$, $q(x,y,t_3):={\bf q}(x,y,t_3)$, (\ref{DSL1M3r}) reads
\begin{eqnarray}\nonumber
&& \alpha_3{q}_{t_3}-c{q}_{xxx}-c_2{q}_{yyy}+
3c\mu{q}_x\int|q|^2_xdy+ 3c_2\mu{q}_y\int|q|^2_ydx+\nonumber \\
&&+3c_2\mu{q}\int({\bar{q}}q_y)_ydx+3c\mu{q}\int({q}_xq)_xdy=0.
\end{eqnarray}
In the real case ${\bf q}=\bar{{\bf q}}$, (\ref{DSL1M3r}) becomes
\begin{equation}\label{DSL1M3rd}
\begin{array}{l}
\alpha_3{\bf q}_{t_3}-c{\bf q}_{xxx}-c_2\beta^3{\bf q}_{yyy}-
cv_1{\bf q}_x+ 3c_2\beta{\bf q}_y{\cal M}_0S_1+ 3c_2{\bf q}{\cal
M}_0S_{1y}-
\\-cv_3{\bf q}-3c_2{\bf q}{\cal M}_0S_2  - 3c_2{\bf q}D^{-1}\left\{{\cal M}_0{\bf q}^{{\top}}{\bf q}{\cal
M}_0S_1-{\cal
M}_0S_1{\cal M}_0{\bf q}^{{\top}}{\bf q}\right\}=0, \\
\beta v_{1y}=-3({\bf q}{\cal M}_0{\bf q}^{\top})_x,\,
S_{1x}=\beta({\bf q}^{\top}{\bf{q}})_y,\,\\ \beta v_{0y}=-3({\bf
q}_x{\cal M}_0{\bf q}^{\top})_x+[{\bf q}{\cal M}_0{\bf
q}^{\top},v_1],\,S_{2x}=\beta^2({\bf q}^{\top}_y{\bf q})_y.
\end{array}
\end{equation}
In the scalar case ($N=m=1$), writing ${\cal M}_0=\mu$ and
$q={q(x,y,t_3)={\bf q}(x,y,t_3)}$,  after setting $y=x$ and
$c+c_2=-1$, $\beta=1$ (\ref{DSL1M3rd}) takes the form
\begin{equation}\label{DSL1M3rdddd}
\alpha_3{ q}_{t_3}+{q}_{xxx}- 6\mu q^2q_x=0,
\end{equation}
which is the mKdV equation. The systems (\ref{DSL1M3r}) and
(\ref{DSL1M3rd}) are therefore, respectively, complex and real,
spatially two-dimensional matrix generalizations of it.

\item $\beta=1$, ${\cal M}_0{\bf r}^{\top}=\nu$ with a constant matrix $\nu$.
In terms of $u := {\bf q}\nu$, (\ref{DSL1M3}) takes the form
\begin{eqnarray}\label{DSL1M3r2}
&& \alpha_3u_{t_3}-cu_{xxx}-c_2u_{yyy}+3cD\left\{\left(\int u_x
dy\right)u\right\}+3c_2\partial_y\left\{u\left(\int u_y
dx\right)\right\} \nonumber \\
&& - c\left(\int[u,v_1]dy\right)u-3c_2u\left(\int[u,S_1]dx\right)=0, \nonumber \\
&&  \nu\left(c\int[u,v_1]dy-3c_2\int[S_1,u]dx\right)=0, \quad
  v_{1y}=-3u_x,\,\,\, S_{1x}=u_y.
\end{eqnarray}
In the scalar case ($N=1,m=1$), this reduces to
\begin{equation}\label{DSL1M3r2scal}
\alpha_3u_{t_3}-cu_{xxx}-c_2u_{yyy}+3cD\left\{\left(\int u_x
dy\right)u\right\}+3c_2\partial_y\left\{u\left(\int u_y
dx\right)\right\}=0,
\end{equation}
which is the Nizhnik equation \cite{Nizhnik80}. The system
(\ref{DSL1M3r2}) thus generalizes the latter to the matrix case.
\end{enumerate}
\end{enumerate}

II. $k=2$

Now we will consider two Lax pairs connected with the operator
$L_2$:
\begin{equation}
L_2=\beta_2\partial_{\tau_2}-D^2-2u-{\bf q}{\cal M}_0D^{-1}{\bf
r}^{\top}.
\end{equation}

\begin{enumerate}
\item $n=2$, $l=1$
\begin{equation}\label{YOP}
\begin{array}{l}
\tilde{M}_{2,1}=M_{2,1}-c_1(L_2)^2=\alpha_2\partial_{t_2}-D^2-2u-c_1\left(\beta_2^2\partial_{\tau_2}^2-2\beta_2\partial_{\tau_2}D^2+\right.\\\left.+D^4-2\beta_2u_{\tau_2}-4\beta_2u\partial_{\tau_2}+4uD^2+4u_xD+2u_{xx}\right.\left.+2{\bf
q}_x{\cal M}_0{\bf r}^{\top}\right.+\\+\left.2{\bf q}{\cal M}_0{\bf
r}^{\top}D+4u^2-2\beta_2{\bf q}{\cal M}_0D^{-1}\partial_{\tau_2}{\bf
r}^{\top}\right).
\end{array}
\end{equation}
Lax equation $[L_2,\tilde{M}_{2,1}]=0$ is equivalent to the system:
\begin{equation}\label{YO}
\left.
\begin{array}{l}
\alpha_2u_{t_2}-\beta_2u_{\tau_2}=({\bf q}{\cal M}_0{\bf
r}^{\top})_x-c_1(L_2\{{\bf q}\}{\cal M}_0{\bf r}^{\top}+{\bf q}{\cal M}_0(L_2^{\tau}\{{\bf r}\})^{\top})_x,\\
M_{2,1}\{{\bf q}\}=0, M_{2,1}^{\tau}\{{\bf r}\}=0.
\end{array}\right.
\end{equation}
In vector case ($N=1$) and under additional reductions $c_1=0$,
$\alpha_2,\beta_2\in i{\mathbb{R}}$, ${\cal M}_0=-{\cal M}_0^*$,
$u=\bar{u}$ and ${\bf q}=\bar{\bf r}$ system (\ref{YO}) reduces to
(2+1)-generalization of Yajima-Oikawa equation
(\ref{Yajima-Oikawa2+1}).
\item $n=3$, $l=1$.

\begin{equation}\label{YOx}
\begin{array}{l}
\tilde{M}_{3,1}\!=\!M_{3,1}\!-\!c_1(L_2)^2=\alpha_3\partial_{t_3}-D^3-3uD-\frac{3}{2}\left(u_x+\beta_2D^{-1}\{u_{\tau_2}\}+{\bf
q}{\cal M}_0{\bf
r}^{\top}\right)-\\-c_1\left(\beta_2^2\partial_{\tau_2}^2-2\beta_2\partial_{\tau_2}D^2+D^4-2\beta_2u_{\tau_2}-4\beta_2u\partial_{\tau_2}+4uD^2+4u_xD+2u_{xx}\right.+\\\left.+2{\bf
q}_x{\cal M}_0{\bf r}^{\top}+2{\bf q}{\cal M}_0{\bf
r}^{\top}D+4u^2-2\beta_2{\bf q}{\cal M}_0D^{-1}\partial_{\tau_2}{\bf
r}^{\top}\right).
\end{array}
\end{equation}
The equation $[L_2,\tilde{M}_{3,1}]=0$ is equivalent to the
following system:
\begin{equation}\label{Yao23}
\left.
\begin{array}{l}
\alpha_3 u_{t_3}-\frac14 u_{xxx}-3uu_x-\frac34\beta_2({\bf q}{\cal
M}_0{\bf r}^{\top})_{\tau_2}+\\+\frac{3}{4}({\bf q}{\cal M}_0{\bf
r}^{\top}_x-{\bf q}_x{\cal M}_0{\bf
r}^{\top})_x+\frac34[u,u_x+\beta_2D^{-1}\{u_{\tau_2}\}+{\bf q}{\cal
M}_0{\bf r}^{\top}]-\\-[{\bf q}{\cal M}_0{\bf
r}^{\top},u]-c_1L_2\{{\bf q}\}{\cal M}_0{\bf r}^{\top}-c_1{\bf
q}{\cal M}_0(L_2^{\tau}\{{\bf
r}\})^{\top}=\frac34\beta_2^2D^{-1}\{u_{\tau_2\tau_2}\},\\\tilde{M}_{3,1}\{{\bf
q}\}=0,\,\,\tilde{M}_{3,1}^{\tau}\{{\bf r}\}=0.
\end{array}\right.
\end{equation}
In the vector case ($N=1$) under additional reductions
$\alpha_3\in{\mathbb{R}}$, $\beta_2\in i{\mathbb{R}}$, ${\cal
M}_0=-{\cal M}_0^*$, ${\bf q}=\bar{\bf r}$ and $u=\bar{u}$ system
(\ref{Yao23}) reduces to (2+1)-generalization of higher
Yajima-Oikawa system (\ref{3100}).
\end{enumerate}
\begin{remark}
After putting $\alpha_2=0$ and $\alpha_3=0$ in formulae
(\ref{1eq})-(\ref{Yao23}) we obtain examples of a new
(1+1)-dimensional Matrix k-cKP hierarchy \cite{NEW} (see case 2
after the proof of Theorem \ref{T1}).
\end{remark}

\section{Dressing methods for Extensions of (2+1)-dimensional k-constrained KP
hierarchy}\label{dressed} In this section our aim is to consider
hierarchy of equations given by the Lax pair (\ref{ex2+1}).
We suppose that the operators $L_k$ and $M_{n,l}$ in (\ref{ex2+1})
satisfy the commutator equation $[L_k,M_{n,l}]=0$. At first we
recall some results from \cite{K2009}. Let $N\times K$-matrix
functions $\varphi$ and $\psi$ be solutions of linear problems:
\begin{equation}\label{pr}
\begin{array}{c}
L_k\{\varphi\}=\varphi\Lambda,\,\,L_k^{\tau}\{\psi\}=\psi\tilde{\Lambda},\,\,\Lambda,\tilde{\Lambda}\in Mat_{K\times K}({\mathbb{C}}).\\
\end{array}
\end{equation}
Introduce  binary Darboux transformation (BDT) in the following way:
\begin{equation}\label{W}
W=I-\varphi\left(C+D^{-1}\{\psi^{\top}\varphi\}\right)^{-1}D^{-1}\psi^{\top},
\end{equation}
where $C$ is a $K\times K$-constant nondegenerate matrix. The
inverse operator $W^{-1}$ has the form:
\begin{equation}\label{W-}
W^{-1}=I+\varphi
D^{-1}\left(C+D^{-1}\{\psi^{\top}\varphi\}\right)^{-1}\psi^{\top}.
\end{equation}
 The following theorem is proven in \cite{K2009}.
\begin{theorem}{\cite{K2009}}\label{2009}
The operator $\hat{L}_k:=WL_kW^{-1}$  obtained from $L_k$ in
(\ref{ex2+1}) via BDT (\ref{W}) has the form
\begin{equation}\label{Lop}
\hat{L}_k:=WL_kW^{-1}=\beta_k\partial_{\tau_k}-\hat{B}_k-\hat{\bf
q}{\cal M}_0D^{-1}{\hat{{\bf r}}}^{\top}+\Phi{\cal
M}_1D^{-1}\Psi^{\top},\, \hat{B}_k=\sum_{j=0}^{k}\hat{u}_jD^j,
\end{equation}
where
\begin{equation}\label{DSM}
\begin{array}{l}
{\cal M}_1=C\Lambda-\tilde{\Lambda}^{\top}C,\,
\Phi=\varphi\Delta^{-1},\,\,
\Psi=\psi\Delta^{-1,\top},\,\Delta=C+D^{-1}\{\psi^{\top}\varphi\},\\
{\hat{\bf q}}=W\{{\bf q}\},\,\,{\hat{\bf r}}=W^{-1,\tau}\{{\bf r}\}.
\end{array}
\end{equation}
 $\hat{u}_j$ are $N\times N$-matrix coefficients depending on functions $\varphi$,
$\psi$ and $u_j$. In particular,
\begin{equation}
\hat{u}_n=u_n,\,\,\hat{u}_{n-1}=u_{n-1}+\left[u_n,\varphi\left(C+D^{-1}\{\psi^{\top}\varphi\}\right)^{-1}\psi^{\top}\right].
\end{equation}
\end{theorem}
Exact forms of all coefficients  $\hat{u}_j$ are given in
\cite{K2009}.

The following corollary follows from Theorem \ref{2009}:
\begin{corollary}
The functions $\Phi=\varphi\Delta^{-1}=W\{\varphi\}C^{-1}$ and $\Psi
=\psi\Delta^{-1,\top}=W^{-1,\tau}\{\psi\}C^{\top,-1}$ satisfy the
equations
\begin{equation}\label{ro}
\hat{L}_k\{\Phi\}=\Phi C\Lambda C^{-1},\,\,
\hat{L}_k^{\tau}\{\Psi\}=\Psi C^{\top}\tilde{\Lambda}C^{\top,-1}.
\end{equation}
\end{corollary}
 For further purposes we will need the following lemmas.
\begin{lemma}
Let ${\cal M}_{l+1}$ be a matrix of the form
\begin{equation}
{\cal M}_{l+1}=C\Lambda^{l+1}-(\tilde{\Lambda}^{\top})^{l+1}C,\,\,
l\in{\mathbb{N}}.
\end{equation}
The following formula holds:
\begin{equation}\label{lemma11}
{\cal M}_{l+1}=\sum_{s=0}^lC\Lambda^sC^{-1}{\cal
M}_1C^{-1}(\tilde{\Lambda}^{\top})^{l-s}C.
\end{equation}
\end{lemma}

\begin{lemma}
The following formula
\begin{equation}
\Phi{\cal M}_{l+1}D^{-1}\Psi^{\top}=\sum_{s=0}^{l}\Phi[s]{\cal
M}_{1}D^{-1}\Psi^{\top}[l-s],
\end{equation}
holds, where
\begin{equation}\label{fk}
\Phi[j]:=(\hat{L}_k)^j\{\Phi\},\,\Psi[j]:=(\hat{L}_k^{\tau})^j\{\Psi\}.
\end{equation}
\end{lemma}
\begin{proof}
Lemma 2 is a consequence of Corollary 1 and formula (\ref{lemma11})
of Lemma 1. Namely, the following relations hold:
\begin{equation}\nonumber
\Phi{\cal
M}_{l+1}D^{-1}\Psi^{\top}=\sum_{s=0}^l{\Phi}C\Lambda^sC^{-1}{\cal
M}_1C^{-1}D^{-1}(\tilde{\Lambda}^{\top})^{l-s}C\Psi^{\top}=\sum_{s=0}^{l}\Phi[s]{\cal
M}_{1}D^{-1}\Psi^{\top}[l-s].
\end{equation}
\end{proof}
Now we assume that the functions $\varphi$ and $\psi$ in addition to
equations (\ref{pr}) satisfy the equations:
\begin{equation}\label{prm}
M_{n,l}\{\varphi\}=c_l\varphi\Lambda^{l+1}=c_lL_k^{l+1}\{\varphi\},\,\,M_{n,l}^{\tau}\{\psi\}=c_l\psi\tilde{\Lambda}^{l+1}=c_l(L_k^{\tau})^{l+1}\{\psi\}.
\end{equation}
Problems (\ref{prm}) can be rewritten via the operator
$\tilde{M}_{n,l}$ (\ref{tM}) as:
\begin{equation}
\tilde{M}_{n,l}\{\varphi\}=0,\,\,\,
\tilde{M}_{n,l}^{\tau}\{\psi\}=0.
\end{equation}
 The following theorem for the operators $M_{n,l}$ (\ref{ex2+1}) and
 $\tilde{M}_{n,l}$ (\ref{tM}) holds:
\begin{theorem}\label{M}
Let $N\times K$ -matrix functions $\varphi$, $\psi$ be solutions of
problems (\ref{pr}) and (\ref{prm}). The transformed operator
$\hat{M}_{n,l}:=WM_{n,l}W^{-1}$ obtained via BDT $W$ (\ref{W}) has
the form:
\begin{equation}\label{Mop}
\begin{array}{l}
\hat{M}_{n,l}:=WM_{n,l}W^{-1}=\alpha_n\partial_{t_n}-\hat{A}_n-c_l\sum_{j=0}^l\hat{{\bf
q}}[j]{\cal M}_0D^{-1}\hat{{\bf
r}}^{\top}[l-j]+\\+c_l\sum_{s=0}^{l}\Phi[s]{\cal
M}_{1}D^{-1}\Psi^{\top}[l-s],\,\,\hat{A}_n=\sum_{i=0}^{n}\hat{v}_iD^i,
\end{array}
\end{equation}
where the matrix ${\cal M}_n$ and the functions $\hat{{\bf q}}$,
$\hat{{\bf r}}$, $\Phi[s]$, $\Psi[l-s]$ are defined by formulae
(\ref{DSM}), (\ref{fk}) and $\hat{{\bf q}}[j]$, $\hat{{\bf r}[j]}$
have the form
\begin{equation}
\hat{{\bf q}}[j]=(\hat{L}_k^j)\{\hat{\bf q}\},\,\,\, \hat{{\bf
r}}[j]=(\hat{L}_k^{j})^{\tau}\{{\hat{\bf r}}\},
\end{equation}
 $\hat{v}_i$ are
$N\times N$-matrix coefficients that depend on the functions
$\varphi$, $\psi$ and $v_i$. The transformed operator
${\hat{\tilde{M}}}_{n,l}=W\tilde{M}_{n,l}W^{-1}$ has the form:
\begin{equation}\label{tmh}
\hat{{\tilde{M}}}_{n,l}=W\tilde{M}_{n,l}W^{-1}=\hat{M}_{n,l}-c_l(\hat{L}_k)^{l+1},
\end{equation}
where $\hat{L}_k$ is given by (\ref{Lop}).
\end{theorem}
\begin{proof}
We shall rewrite the operator $M_{n,l}$ (\ref{ex2+1}) in the form
\begin{equation}\label{Mn1ns}
M_{n,l}=\alpha_n\partial_{t_n}-\sum_{i=0}^{n}v_iD^i-c_l\tilde{{\bf
q}}\tilde{{\cal M}}_0D^{-1}\tilde{\bf r}^{\top},
\end{equation}
where $\tilde{{\cal M}}_0$ is an $m(l+1)\times m(l+1)$-
block-diagonal matrix with entries of ${\cal M}_0$ at the diagonal;
$\tilde{{\bf q}}:=({\bf q}[0],{\bf q}[1],\ldots,{\bf q}[l])$,
$\tilde{{\bf r}}:=({\bf r}[l],{\bf r}[l-1],\ldots,{\bf r}[0])$.
Using Theorem \ref{2009} we obtain that
\begin{equation}\label{fgre}
\hat{M}_{n,l}=\alpha_n\partial_{t_n}-\sum_{i=0}^{n}\hat{v}_iD^i-c_l\hat{\tilde{{\bf
q}}}\tilde{{\cal M}}_0D^{-1}\hat{\tilde{{\bf r}}}^{\top} +\Phi{\cal
M}_{l+1}D^{-1}\Psi^{\top},
 \end{equation}
where $\hat{\tilde{{\bf q}}}=W\{\tilde{{\bf q}}\}$,
$\hat{\tilde{{\bf q}}}=W^{-1,\tau}\{{\tilde{\bf r}}\}$. Using the
exact form of ${\tilde {\bf q}}$ and $\tilde{\bf r}$ we have
\begin{equation}
\begin{array}{l}
\hat{\tilde{{\bf q}}}=W\{\tilde{{{\bf q}}}\}=(W\{{\bf
q}[0]\},\ldots, W\{{\bf q}[l]\}),\,\\ \hat{\tilde{{\bf
r}}}=W^{-1,\tau}\{\tilde{{{\bf r}}}\}=(W^{-1,\tau}\{{\bf
r}[l]\},\ldots, W^{-1,\tau}\{{\bf r}[0]\}).
\end{array}
\end{equation}
We observe that
\begin{equation}
W\{{\bf q}[i]\}=WL^i\{{\bf q}\}=WL^iW^{-1}\{W\{{\bf
q}\}\}=\hat{L}^i\{\hat{{\bf q}}\}=:\hat{{\bf q}}[i].
\end{equation}
It can be shown analogously that
 $W^{-1,\tau}\{{\bf
r}[i]\}=\hat{L}^{\tau,i}\{W^{-1,\tau}\{{\bf
r}\}\}=\hat{L}^{\tau,i}\{\hat{\bf r}\}=:\hat{\bf r}[i]$. Thus we
have:
\begin{equation}\label{sd}
 \hat{\tilde{{\bf
q}}}\tilde{{\cal M}}_0D^{-1}\hat{\tilde{{\bf
r}}}^{\top}=\sum_{j=0}^l\hat{{\bf q}}[j]{\cal M}_0D^{-1}\hat{{\bf
r}}^{\top}[l-j].
\end{equation}
For the last item in (\ref{fgre}) from Lemma 2 we have:
\begin{equation}\label{sdd}
\Phi{\cal M}_{l+1}D^{-1}\Psi^{\top}=\sum_{s=0}^{l}\Phi[s]{\cal
M}_{1}D^{-1}\Psi^{\top}[l-s].
\end{equation}
Using formulae (\ref{fgre}), (\ref{sd}), (\ref{sdd}) we obtain that
the operator $\hat{M}_{n,l}$ has form (\ref{Mop}). The exact form of
the operator $\hat{\tilde{M}}_{n,l}$ follows from formula
(\ref{Mop}) and Theorem \ref{2009}.
\end{proof}
From Theorem \ref{M} we obtain the following corollary.
\begin{corollary}\label{Corol}
Assume that functions $\varphi$ and $\psi$ satisfy problems
(\ref{pr}) and (\ref{prm}). Then the functions
$\Phi=W\{\varphi\}C^{-1}$ and $\Psi=W^{-1,\tau}\{\psi\}C^{\top,-1}$
(see formulae (\ref{DSM})) satisfy the equations:
\begin{equation}\label{meq}
\hat{\tilde{M}}_{n,l}\{\Phi\}=\hat{M}_{n,l}\{\Phi\}-c_l(\hat{L}_k)^{l+1}\{\Phi\}=0,
\,\,\hat{\tilde{M}}_{n,l}^{\tau}\{\Psi\}=\hat{M}^{\tau}_{n,l}\{\Psi\}-c_l(\hat{L}^{\tau}_k)^{l+1}\{\Psi\}=0,
\end{equation}
where the operators $\hat{L}_k$, $\hat{M}_{n,l}$ and
$\hat{\tilde{M}}_{n,l}$ are defined by (\ref{Lop}), (\ref{Mop}) and
(\ref{tmh}).
\end{corollary}
As it was shown in Sections \ref{kckp}-\ref{extended} the most
interesting systems arise from the (2+1)-dimensional k-cKP hierarchy
(\ref{spa1}) and its extension (\ref{ex2+1}) after a Hermitian
conjugation reduction. Our aim is to show that under additional
restrictions Binary Darboux Transformation $W$ (\ref{W}) preserves
this reduction. We shall point out that a differential dressing
operator that was used in \cite{MSS,6SSS,LZL2} does not satisfy such
a property. It imposes nontrivial constraints on the dressing
operator (see, for example, \cite{BS2}).
The following proposition holds.
\begin{proposition}\label{her}

\begin{enumerate}
\item
Let $\psi=\bar{\varphi}$ and $C=C^*$ in the dressing operator $W$
(\ref{W}). Then the operator $W$ is unitary ($W^*=W^{-1}$).
\item
Let the operator $L_k$ (\ref{ex2+1}) be Hermitian (skew-Hermitian)
and $M_{n,l}$ (\ref{ex2+1}) be Hermitian (skew-Hermitian). Then the
operator $\hat{L}_k=WL_kW^{-1}$ (see (\ref{Lop})) transformed via
the unitary operator $W$ is Hermitian (skew-Hermitian) and
$\hat{M}_{n,l}:=WM_{n,l}W^{-1}$ (\ref{tmh}) is Hermitian
(skew-Hermitian).
\item Assume that the conditions of items 1 and 2 hold. Let
$\tilde{\Lambda}=\bar{\Lambda}$ in the case of Hermitian $L_k$
($\tilde{\Lambda}=-\bar{\Lambda}$ in skew-Hermitian case). We shall
also assume that the function $\varphi$ satisfies the corresponding
equations in formulae (\ref{pr}) and (\ref{prm}). Then ${\cal
M}_1=-{\cal M}_1^*$ (${\cal M}_1={\cal M}_1^*$) and
$\Psi=\bar{\Phi}$ (see formulae (\ref{DSM})).
\end{enumerate}
\end{proposition}
\begin{proof}
By using formulae (\ref{W}) and (\ref{W-}) it is easy to check that
$W$ is unitary (in the case ${\psi}=\bar{\varphi}$, $C=C^*$). Assume
that $L_k=L_k^*$ (in the case of  skew-Hermitian $L_k$ analogous
considerations can be applied). Then we have
$(\hat{L}_k)^*=W^{-1,*}L^*_kW^{*}=WL_kW^{-1}=\hat{L}_k$. Analogously
we can show that $W$ maintains Hermitian (skew-Hermitian) property
for $M_{n,l}$. The formulae from item 3 can be checked by direct
calculations.
\end{proof}
 Now we list several examples of dressing for Lax pairs and
the corresponding equations from Section \ref{extended}.

Consider dressing methods for equations connected with the operator
$L_0$ (\ref{L0}). Assume that $\varphi$ and $\psi$ are $N\times
K$-matrix functions that satisfy the equations
\begin{equation}\label{Leq1}
 L_0\{\varphi\}=\varphi\Lambda,\,\,L_0^{\tau}\{\psi\}=\psi\tilde{\Lambda},\,\,L_0:=\beta\partial_y.
 \end{equation}
 By Theorem \ref{2009} we obtain that the dressed operator $\hat{L}_0$
 via BDT $W$ (\ref{W}) has the form
 \begin{equation}\label{hL0}
 \hat{L}_0=WL_0W^{-1}=\beta\partial_y+\Phi{\cal
 M}_1D^{-1}\Psi^{\top}.
 \end{equation}
 \begin{enumerate}
\item $n=2$, $l=1$.
Assume that $N\times K$-matrix functions $\varphi$ and $\psi$ in
addition to equations (\ref{Leq1}) also satisfy the equations
\begin{equation}\label{M2e}
 M_2\{\varphi\}=c_1\varphi\Lambda^2=c_1L_0^2\{\varphi\},\,\,M_2^{\tau}\{\psi\}=c_1\psi{\tilde\Lambda}^2=c_1(L_0^{\tau})^2\{\psi\},\,
 M_2:=\alpha_2\partial_{t_2}-D^2.
\end{equation}
By Theorem \ref{M} we obtain that the transformed operator
$\hat{M}_2$ has the form
\begin{equation}
\hat{M}_2=WM_2W^{-1}=\alpha_2\partial_{t_2}-D^2-\hat{v}_0+\hat{L}_0\{\Phi\}{\cal
 M}_1D^{-1}\Psi^{\top}+\Phi{\cal
 M}_1D^{-1}((\hat{L}_0^{\tau})\{\Psi\})^{\top}.
\end{equation}
By direct calculations it can be obtained that
$\hat{v}_0=2(\varphi\Delta^{-1}\psi^{\top})_x$,
$\Delta=C+D^{-1}\{\psi^{\top}\varphi\}$. It can be easily checked
that
\begin{equation}\label{a}
\begin{array}{l}
\beta(\varphi\Delta^{-1}\psi^{\top})_y=\beta\varphi_y\Delta^{-1}\psi^{\top}-\beta\varphi\Delta^{-1}D^{-1}\{\psi^{\top}\varphi\}_y\Delta^{-1}\psi^{\top}+\beta\varphi\Delta^{-1}\psi^{\top}_y=\\
=\varphi\Lambda\Delta^{-1}\psi^{\top}-\beta\varphi\Delta^{-1}D^{-1}\{\psi^{\top}\varphi\}_y\Delta^{-1}\psi^{\top}-\varphi\Delta^{-1}\tilde{\Lambda}^{\top}\psi^{\top}
=\\= \varphi\Delta^{-1}(C\Lambda+
D^{-1}\{\psi^{\top}\varphi\}\Lambda)\Delta^{-1}\psi^{\top}-\beta\varphi\Delta^{-1}D^{-1}\{\psi^{\top}\varphi\}_y\Delta^{-1}\psi^{\top}+\\+\varphi\Delta^{-1}
(-\tilde{\Lambda}^{\top}C-
\tilde{\Lambda}^{\top}D^{-1}\{\psi^{\top}\varphi\})\Delta^{-1}\psi^{\top}=
\\= \varphi\Delta^{-1}(C\Lambda+\beta
D^{-1}\{\psi^{\top}\varphi_y\})\Delta^{-1}\psi^{\top}-\beta\varphi\Delta^{-1}D^{-1}\{\psi^{\top}\varphi\}_y\Delta^{-1}\psi^{\top}+\\+\varphi\Delta^{-1}
(-\tilde{\Lambda}^{\top}C+\beta
D^{-1}\{\psi^{\top}_y\varphi\})\Delta^{-1}\psi^{\top}=\Phi {\cal
M}_1\Psi^{\top}.
\end{array}
\end{equation}
From the latter formula we obtain that
\begin{equation}\label{adS}
\beta\hat{v}_{0y}=2\beta(\varphi\Delta^{-1}\psi^{\top})_{xy}=2(\Phi{\cal
M}_1\Psi^{\top})_x.
\end{equation}
From Corollary \ref{Corol} we see that the functions
$\Phi=\varphi\Delta^{-1}$ and $\Psi$$=\psi\Delta^{\top,-1}$ where
$\Delta=C+D^{-1}\{\psi^{\top}\varphi\}$ (see formulae (\ref{DSM}))
 satisfy equations (\ref{meq}). After the change ${\bf q}:=\Phi$,
${\bf r}:=\Psi$, ${\cal M}_0:=-{\cal M}_1$, $v_0:=\hat{v}_0$ from
formulae (\ref{meq}) and (\ref{adS}) we obtain that $N\times
K$-matrix functions ${\bf q}$, ${\bf r}$, an $N\times N$-matrix
function $v_0$ and a $K\times K$-matrix ${\cal M}_0$ satisfy
equations (\ref{DSnr}) in the case $c=1$. In the case of additional
reductions in formulae (\ref{Leq1})-(\ref{M2e}): $\alpha_2\in
i{\mathbb{R}}$, $\beta\in{\mathbb{R}}$, $c_1\in{\mathbb{R}}$,
$\tilde{\Lambda}=-\bar{\Lambda}$, $\psi=\bar{\varphi}$ and $C=C^*$
in gauge transformation operator $W$ (\ref{W}) from Proposition
\ref{her} we obtain that the functions ${\bf q}:=\Phi$ and
$v_0=\hat{v}_0=2(\varphi\Delta^{-1}\varphi^*)_x$ satisfy matrix DS
system (\ref{DS}) in the case $c=1$.

Consider more precisely the exact solutions of a vector version
($N=1$) of DS system (\ref{DS}) in the case $c=1$.
 Assume that the $K\times K$-matrix
$\Lambda$ in (\ref{Leq1}) is diagonal:
$\Lambda={\rm{diag}}(\lambda_1,\ldots,\lambda_K)$,
$\lambda_j\in{\mathbb{C}}$. Let us fix arbitrary natural numbers
$K_1$, $\ldots$, $K_m$. Denote by $K$ the number $K=K_1+\ldots+K_m$.
 Our aim is to present the exact form of solution of $m$-component DS system (\ref{DS}) (${\bf q}=(q_1,\ldots
q_m)$) 
Assume that a $K\times K$-matrix $C$ has the form:
\begin{equation}\label{C}
C={\rm{diag}}(C_{K_1},\ldots,C_{K_m}),\,\,
C_{K_s}=\left(\mu_s\frac{1}{\lambda_j+\bar{\lambda}_i}\right)_{i,j=1}^{K_s},\,\,\mu_s=\pm1,\,\,s=\overline{1,m}.
\end{equation}
I.e., the matrix $C$ has diagonal blocks $C_{K_s}$ of dimension
$K_s$ on its diagonal.
$\mu_s$, $s=\overline{1,m}$, are numbers each equal to $1$ or $-1$.
It can be checked by direct calculations that the matrix ${\cal
M}_1:=C\Lambda+\Lambda^*C$ has the form:
\begin{equation}\label{Mat}
{\cal M}_1=\left(\begin{array}{cccccc}\mu_1E_{K_1}&0_{K_1,K_2}&0_{K_1,K_3}&\ldots&0_{K_1,K_{m-1}}&0_{K_1,K_m}\\
0_{K_2,K_1}&\mu_2E_{K_2}&0_{K_2,K_3}&\ldots&0_{K_2,K_{m-1}}&0_{K_2,K_{m}}\\\
0_{K_3,K_1}&0_{K_3,K_2}&\mu_3E_{K_3}&\ldots&,0_{K_3,K_{m-1}}&0_{K_3,K_{m}}\\
\vdots&\vdots&\vdots&\ldots&\vdots&\vdots\\
0_{K_{m-1},K_1}&0_{K_{m-1},K_2}&0_{K_{m-1},K_3}&\ldots&\mu_{m-1}E_{K_{m-1}}&0_{K_{m-1},K_{m}}\\\
0_{K_m,K_1}&0_{K_m,K_2}&0_{K_m,K_3}&\ldots&
0_{K_m,K_{m-1}}&\mu_mE_{K_m}
\end{array}\right),
\end{equation}
where by $E_{K_s}$ we denote the $K_s\times K_s$-square matrix
consisting of 1. I.e., $E_{K_s}={\bf 1}_{K_s}^{\top}{\bf1}_{K_s}$,
where  ${\bf 1}_{K_s}=(1,\ldots,1)$ is $1\times K_s$-vector
consisting of $1$. $0_{K_i,K_j}$ is the $K_i\times K_j$-matrix
consisting of zeros.  The matrix ${\cal M}_1$ can be rewritten as
${\cal M}_1={\rm{diag}}(\mu_1{\bf 1}_{K_1}^{\top}{\bf
1}_{K_1},\ldots,\mu_m{\bf 1}_{K_m}^{\top}{\bf 1}_{K_m})$. It can be
noticed that the matrix ${\cal M}_1$ admits the factorization:
\begin{equation}\label{factor}
\begin{array}{l} {\cal M}_1=P\sigma P^*,\,\,\
P=(e_{K_1},\ldots,e_{K_m}),\sigma={\rm{diag}}(\mu_1,\ldots,\mu_m),\\
e_{K_j}=(0_{K_1},\ldots, 0_{K_{j-1}},{\bf
1}_{K_j},0_{K_{j-1}},\ldots,0_{K_l})^{\top}.
\end{array}
\end{equation}

In formula (\ref{factor}) by $0_{K_j}=(0,\ldots,0)$ we denote the
$1\times K_j$-vector consisting of zeros. $e_{K_j}$ denotes the
$1\times K$-vector consisting of row vectors $0_{K_j}$ and ${\bf
1}_{K_j}$.
\begin{Example}
Consider the case $m=2$, $K_1=1$, $K_2=2$. Then formulae
(\ref{Mat})-(\ref{factor}) for the matrix ${\cal M}_1$ become
\begin{equation}\label{Mat12}
{\cal M}_1=\left(\begin{array}{ccc}\mu_1&0&0\\
0&\mu_2&\mu_2\\\
0&\mu_2&\mu_2
\end{array}\right),\,\,\,{\cal M}_1=P\sigma P^*=\left(\begin{array}{ccc}1&0\\
0&1\\\
0&1
\end{array}\right)\left(\begin{array}{ccc}\mu_1&0\\
0&\mu_2
\end{array}\right)\left(\begin{array}{ccc}1&0&0\\
0&1&1
\end{array}\right).
\end{equation}
\end{Example}

We put $\beta=1$ in (\ref{Leq1}) and choose solution of systems
(\ref{Leq1}) and (\ref{M2e}) in the following form:
\begin{equation}
\varphi=(\varphi_1,\ldots\varphi_K),\,\,
\varphi_j=\exp{\left\{\left(\frac{\nu_j^2+c_1{\lambda}^2_j}{\alpha_2}\right)t_2+\nu_jx+{\lambda}_jy\right\}},\,\
\lambda_j,\nu_j\in{\mathbb{C}}.
\end{equation}
Using Corollary \ref{Corol}, and formula (\ref{adS}) we see that the
functions $\Phi$, $v_0=2(\varphi\Delta^{-1}\varphi^*)_x$,
$\hat{S}=-2D^{-1}\{\Phi^*\Phi\}_y=2(\Delta^{-1})_y $ and ${\cal
M}_1=P\sigma P^*$ satisfy the matrix DS-system (\ref{DS}):
\begin{equation}\label{DSq1}
\begin{array}{c}
 \alpha_2\Phi_{t_2}=\Phi_{xx}+c_1{\Phi}_{yy}+v_0\Phi-c_1\Phi P\sigma P^*\hat{S},\,\, \alpha_2\in i{\mathbb{R}},\\
 v_{0y}=2(\Phi P\sigma P^*\Phi^*)_x,\,\,\hat{S}_{x}=-2(\Phi^*\Phi)_y.
\end{array}
\end{equation}

 Define functions ${\bf q}$, $v_0$ and $S$ in the following way:
\begin{equation}\label{qP}
\begin{array}{l}
{\bf q}\!=\!\Phi
P=\varphi\Delta^{-1}\!P=\!\varphi\left(C+D^{-1}\{\varphi^*\varphi\}\right)^{-1}\!\!P,\,\,
v_0\!=\!2(\varphi\Delta^{-1}\varphi^*\!)_x,\,\,\\
S=2P^*(\Delta^{-1})_yP,
\end{array}
\end{equation}
 where the
matrices $C$ and $P$ are defined by (\ref{C}) and (\ref{factor}).
The integral $D^{-1}$ with respect to $x$ in the formula
$D^{-1}\{\varphi^*\varphi\}$ (\ref{qP}) is realized in the form
\begin{equation}
D^{-1}\{\varphi^*\varphi\}=\left(\frac{\bar{\varphi}_i\varphi_j}{\nu_j+\bar{\nu}_i}\right)_{i,j=1}^{K}.
\end{equation}
Then from equation (\ref{DSq1}) it follows that functions (\ref{qP})
satisfy the $m$-component DS system
\begin{equation}
\begin{array}{c}
 \alpha_2{\bf q}_{t_2}={\bf q}_{xx}+c_1{\bf q}_{yy}+v_0{\bf q}-c_1{\bf
 q}\sigma S,\,\,\, \alpha_2\in i{\mathbb{R}},\\
v_{0y}=2({\bf q}\sigma{\bf q}^*)_x,\,\,S_{x}=-2({\bf q}^*{\bf q})_y,
\end{array}
\end{equation}
that can be rewritten as:
\begin{equation}\label{DSi}
\begin{array}{c}
 \alpha_2q_{i,t_2}={q}_{i,xx}+c_1{q}_{i,yy}+v_0{q}_i-c_1\sum_{j=1}^m\mu_j{q}_j S_{ji},\,\,i=\overline{1,m},\,\, \alpha_2\in i{\mathbb{R}},\\
 v_{0y}=2(\sum_{j=1}^m\mu_j|q|^2_j)_x,\,\,S_{ji,x}=-2(\bar{q}_jq_i)_y,\,\,i,j=\overline{1,m}.
\end{array}
\end{equation}
\begin{remark}\label{rem2}
By the well-known formulae from the matrix theory solutions
(\ref{qP}) can be rewritten as
\begin{eqnarray}
&&q_j=\Phi
e_{K_j}=\sum_{i=K_{j-1}+1}^{K_{j}}\Phi_i=\varphi\Delta^{-1}e_{K_j}=-\frac{\det\left(\begin{array}{ll}
\Delta & e_{K_j} \\ \varphi & 0
\end{array} \right)}{\det \Delta},\,\,j=\overline{1,m};\label{w3}\\
&&v_0=2(\varphi\Delta^{-1}\varphi)^*_x=-2\left(\frac{\det\left(\begin{array}{ll}
\Delta & \varphi^* \\ \varphi & 0 \end{array} \right)}{\det
\Delta}\right)_x,\,\,\\
&&S_{ij}=2(P_{\cdot i})^*(\Delta^{-1})_yP_{\cdot
j}=-2\left(\frac{\det\left(\begin{array}{ll} \Delta & P_{\cdot j}
\\ (P_{\cdot i})^* & 0
\end{array} \right)}{\det
\Delta}\right)_y,\,\,i,j=\overline{1,m}.\label{w1}
\end{eqnarray}
In formulae (\ref{w3})-(\ref{w1}) we set $K_0:=0$. By  $S_{ij}$,
$i,j=\overline{1,m}$, we denote the elements of the corresponding
matrix-function $S$. By $q_j$ we denote the elements of the
vector-function ${\bf q}=(q_1,\ldots,q_m)$. $P_{\cdot i}$ denotes
the $i$-th vector-column of the matrix $P$.
\end{remark}
In the case $m=1$ system  (\ref{DSi}) reduces to the following one:

\begin{equation}\label{DSi1}
\begin{array}{c}
 \alpha_2q_{t_2}={q}_{xx}+c_1{q}_{yy}+v_0{q}-c_1\mu_1qS,\,\, \alpha_2\in i{\mathbb{R}},\\
 v_{0y}=2(\mu_1|q|^2)_x,\,\,S_{x}=-2(|q|^2)_y,
\end{array}
\end{equation}
and the formulae from Remark \ref{rem2} take the form:
\begin{eqnarray}\nonumber
&&q=\varphi\Delta^{-1}{\bf1}^{\top}_K=-\frac{\det\left(\begin{array}{ll}
\Delta & {\bf1}^\top_K \\ \varphi & 0
\end{array} \right)}{\det \Delta},\,\, S=2({\bf 1}
_K\Delta^{-1}{\bf
1}^{\top}_K)_y=-2\left(\frac{\det\left(\begin{array}{ll} \Delta &
{\bf 1}_K^{\top}
\\ {\bf 1}_K & 0
\end{array} \right)}{\det
\Delta}\right)_y,
\nonumber\\&&v_0=2(\varphi\Delta^{-1}\varphi^*)_x=-2\left(\frac{\det\left(\begin{array}{ll}
\Delta & \varphi^* \\ \varphi & 0 \end{array} \right)}{\det
\Delta}\right)_x.\nonumber
\end{eqnarray}
and represent a $K$-soliton solution of the scalar DS system. The
latter in the case $K=1$  takes the form:
\begin{equation}\label{1sol}
q=\frac{\exp({\theta})}{\Delta},\,\,
S=-2\frac{\rm{Re}(\lambda_1)\exp(2\rm{Re}(\theta))}{\rm{Re}(\nu_1)\Delta^2},\,\,
v_0=2\mu_1\frac{\rm{Re}(\nu_1)\exp(2\rm{Re}(\theta))}{\rm{Re}(\lambda_1)\
\Delta^2},\,
\end{equation}
where
$\Delta=\frac{\mu_1}{2\rm{Re}(\lambda_1)}+\frac{1}{2\rm{Re}({\nu}_1)}
\exp{(2\rm{Re}(\theta))}$ and
$\theta=\left(\frac{\nu_1^2+c_1{\lambda}^2_1}{\alpha_2}\right)t_2+\nu_1x+{\lambda}_1y$.
 We shall point out that  DS equation (\ref{DSi1}) consists of two special cases:
\begin{enumerate}
\item $\mu_1=1$. In this case formulae (\ref{w3})-(\ref{w1}) from Remark \ref{rem2} and (\ref{1sol}) represent regular solutions 
 of DS-equation
(\ref{DSi1}).
\item $\mu_1=-1$. In this case Remark \ref{rem2} (formulae (\ref{w3})-(\ref{w1})) and formula (\ref{1sol}) give us singular solutions  of DS-equation
(\ref{DSi1}).
\end{enumerate}
 In a similar way dressing technique has been elaborated for several
multicomponent integrable systems from the k-cKP hierarchy in
\cite{SCD}.

\begin{remark}\label{RemDrom}
It should be also pointed out that besides soliton solutions DS
equation (\ref{DSi1}) has also an important class of solutions known
as dromions. Dromions were obtained at first in \cite{BLMP} via
B\"acklund transformations and nonlinear superposition formulae.
They were also recovered by the Hirota direct method \cite{JHRH} and
Darboux Transformations \cite{LSY}. In paper \cite{GM} dromions were
constructed for matrix DS equation (\ref{DS}). However, the Lax pair
for DS equation used in \cite{BLMP,LSY} involves non-stationary
Dirac operator and differs from the integro-differential Lax pair
(\ref{L0}), (\ref{1eq}) for DS equation that we considered.
\end{remark}
\item $n=3$, $l=2$.
Assume that in addition to equations (\ref{Leq1}) the functions
$\varphi$ and $\psi$ satisfy the equations:
\begin{equation}\label{M3dr}
\begin{array}{l}
M_3\{\varphi\}=c_2\varphi\Lambda^3=c_2L_0^3\{\varphi\},\,\,M_3^{\tau}\{\psi\}=c_2\psi{\tilde\Lambda}^3=c_2(L_0^{\tau})^3\{\psi\},\,\\
 M_3:=\alpha_3\partial_{t_3}-D^3.
 \end{array}
\end{equation}
As it was done in the case $n=2$, $l=1$, we obtained that the
transformed operator $\hat{M}_3$ via BDT $W$ (\ref{W}) has the form
\begin{equation}\label{form}
\hat{M_3}=WM_3W^{-1}=\alpha_3\partial_{t_3}-D^3-\hat{v}_1D-\hat{v}_0+c_2\sum_{i=0}^2\hat{L}^i_0\{\Phi\}{\cal
M}_0D^{-1}((\hat{L}_0^{\tau})^{2-i}\{\Psi\})^{\top},
\end{equation}
where  $\hat{v}_1=3(\varphi\Delta^{-1}\psi^{\top})_x$,
$\hat{v}_0=3\varphi\Delta^{-1}\psi^{\top}(\varphi\Delta^{-1}\psi^{\top})_x-3(\varphi_x\Delta^{-1}\psi^{\top})_x$.
After the change ${\bf q}:=\Phi$, ${\bf r}:=\Psi$, ${\cal
M}_0:=-{\cal M}_1$, $v_1:=\hat{v}_1$, $v_0:=\hat{v}_0$ it can be
shown that the functions ${\bf q}$, ${\bf r}$, $v_1$, $v_0$ with the
matrix ${\cal M}_0$ satisfy equations (\ref{DSL1M3}) in case $c=1$.

Analogously, dressing methods can be done for the other Lax pairs
from Section \ref{extended}. As an example we consider equation
(\ref{Yao23}) with Lax pair (\ref{YOx}):

\item $k=2,n=3,l=1$.
Assume that $N\times K$-matrix functions $\varphi$ and $\psi$
satisfy the equations:
\begin{equation}\label{L2M3}
\begin{array}{c}
L_2\{\varphi\}=\varphi\Lambda,\,\,L_2^{\tau}\{\psi\}=\psi\tilde{\Lambda},\,\,L_2:=\beta_2\partial_{\tau_2}-D^2,\\
M_3\{\varphi\}=c_1\varphi\Lambda^2,\,\,M_3^{\tau}\{\psi\}=c_1\psi{\tilde{\Lambda}}^2,\,M_3:=\alpha_3\partial_{t_3}-D^3.
\end{array}
 \end{equation}
   Using
Theorem \ref{2009} we obtain that the dressed operator $\hat{L}_2$
has the form:
\begin{equation}
\hat{L}_2=\beta_2\partial_{\tau_2}-D^2-2\hat{u}+{\Phi}{\cal
M}_1D^{-1}{\Psi}^{\top},
\end{equation}
where the functions $\Phi$ and $\Psi$ are defined by (\ref{DSM});
$\hat{u}=(\varphi\Delta^{-1}\psi^{\top})_x$. Using Theorems
\ref{2009} and \ref{M}, their corollaries and direct calculations it
can be checked that the functions
$\hat{u}=(\varphi\Delta^{-1}\psi^{\top})_x$, ${\bf q}:=\Phi$, ${\bf
r}=\Psi$ and the matrix ${\cal M}_0:=-{\cal M}_1$ satisfy equations
(\ref{Yao23}). We note that in the case $c_1=0$, $N=1$ we obtain
solutions for the (2+1)-dimensional generalization of higher
Yajima-Oikawa equation:
\begin{equation}\label{Yao23m}
\left.
\begin{array}{l}
\alpha_3\! u_{t_3}\!-\!\frac14
u_{xxx}\!-\!3uu_x\!-\!\frac34\beta_2({\bf q}{\cal M}_0{\bf
r}^{\top}\!)_{\tau_2}\!+\!\frac{3}{4}({\bf q}{\cal M}_0{\bf
r}^{\top}_x\!-\!{\bf q}_x{\cal M}_0{\bf
r}^{\top})_x\!\!=\!\!\frac34\beta_2^2D^{-1}\{\!u_{\tau_2\tau_2}\!\},
\\
\alpha_3{\bf q}_{t_3}-{\bf q}_{xxx}-3u{\bf
q}_x-\frac{3}{2}\left(u_x+\beta_2D^{-1}\{u_{\tau_2}\}+{\bf q}{\cal
M}_0{\bf r}^{\top}\right){\bf q}=0,\\\ \alpha_3{\bf r}_{t_3}-{\bf
r}_{xxx}-3(u{\bf
r})_x+\frac{3}{2}\left(u_x+\beta_2D^{-1}\{u_{\tau_2}\}+{\bf q}{\cal
M}_0{\bf r}^{\top}\right){\bf r}=0,
\end{array}\right.
\end{equation}
that was investigated as a member of (2+1)-dimensional k-cKP  in
\cite{MSS,6SSS}. Its solutions via differential gauge operators were
considered in \cite{LZL1,LZL2}. It should be mentioned that the
solution generating technique via Binary Darboux Transformation $W$
(\ref{W}) presented for matrix DS-system in item 1 can also be
elaborated for system (\ref{Yao23m}) in the case of Hermitian
conjugation reduction ($u$ is a real-valued function, ${\bf
r}=\bar{\bf q}$, $\alpha_3\in{\mathbb{R}}$, $\beta_2\in
i{\mathbb{R}}$, ${\cal M}^*_0=-{\cal M}_0$).
\end{enumerate}
We also point out that some of the matrix equations presented in
items 1-3 and their solutions were also investigated recently
\cite{DMH,GM}.
\section{Conclusions}
In this paper we proposed a new (2+1)-BDk-cKP hierarchy
(\ref{ex2+1}) that generalizes (2+1)-dimensional KP hierarchy
(\ref{spa1}) which was introduced in \cite{MSS,6SSS} and
rediscovered in \cite{LZL1}. For some members of (2+1)-dimensional
k-cKP hierarchy (\ref{spa1}) (e.g. (\ref{Yajima-Oikawa2+1}) and
(\ref{3100})), solutions were obtained via the binary Darboux
transformation dressing method \cite{BS1,PHD}. Dressing methods for
(2+1)-dimensional extension of the k-cKP and modified k-cKP
hierarchy via differential transformations were elaborated recently in \cite{LZL1}.
 New (2+1)-BDk-cKP hierarchy (\ref{ex2+1}) extends hierarchy
 (\ref{spa1}) \cite{MSS,6SSS,LZL1,LZL2}. It also contains
integrable systems and their matrix generalizations that do not
belong to (2+1)-dimensional k-cKP hierarchy (\ref{ex2+1}).
(2+1)-BDk-cKP hierarchy also extends matrix KP hierarchies such as
$(t_A,\tau_B)$-hierarchy \cite{Zeng} and
$(\gamma_A,\sigma_B)$-hierarchy \cite{YYQB}. Some members of new
hierarchy (\ref{ex2+1}) such as Davey-Stewartson systems
(DS-I,DS-II,DS-III), matrix generalizations of (2+1)-dimensional
extensions of the mKdV and their representations via
integro-differential operators were considered recently in \cite{n}.
Dressing methods via binary Darboux transformations presented in
Section \ref{dressed} give us an opportunity to construct the exact
solutions for equations that are contained in (2+1)-BDk-cKP
hierarchy (\ref{ex2+1}).
The most interesting systems obtained from hierarchy (\ref{ex2+1})
arise after a Hermitian conjugation reduction. It is much more
suitable to use BDT for dressing methods in this situation rather
then a differential operator. Elaborated dressing technique method
provides us also with a possibility to investigate exact solutions
of matrix equations that (2+1)-BDk-cKP hierarchy (\ref{ex2+1})
contains. We shall note that the interest to noncommutative
equations has also arisen recently \cite{DMH,GM,GNS,GHN}. Thus one
of problems for further investigation is the generalization of
(2+1)-BDk-cKP hierarchy to the case of noncommutative algebras
\cite{March,DMH,GM,GNS,GHN} (namely, the elements ${\bf q}$, ${\bf
r}$, $u_j$ and $v_i$ belong to some noncommutative ring). It is also
very important to point out that initial operators that we have
chosen for dressing in (\ref{Leq1}), (\ref{M2e}), (\ref{M3dr}),
(\ref{L2M3}) are differential. However, the statements of the
theorems in Section \ref{dressed} (in particular Theorem \ref{2009}
and \ref{M})) concern initial integro-differential operators. The
case of initial (undressed) integro-differential operators should
also be elaborated. We recently obtained that for (1+1)-dimensional
integrable systems it leads to classes of solutions that do not tend
to zero at infinities (e.g. finite density solutions).
Generalizations of (2+1)-dimensional k-cKP hierarchy (\ref{ex2+1})
also provides us with problems for further investigations. In
particular, using (2+1)-BDk-cKP hierarchy it becomes possible to
investigate generalizations of modified (2+1)-extended k-cKP
hierarchy (modified (2+1)-BDk-cKP hierarchy) and elaborate dressing
methods for it. Some members of this hierarchy were recently
considered in \cite{n}. In particular, we investigated Lax
integro-differential representations for  the Nizhnik equation
\cite{Nizhnik80} and (2+1)-dimensional extension of the Chen-Lee-Liu
equation \cite{Chen}. Representations for some of those systems in
the algebra of purely differential operators with
matrix coefficients
can be found in \cite{11Athorne}. Using (2+1)-BDk-cKP hierarchy we
also intend to generalize (2+1)-extended Harry-Dym hierarchy
presented in \cite{WXM}. We recently obtained multi-component
equation with the following Lax pair contained in (2+1)-BDk-cmKP
hierarchy (bidirectional generalization of (2+1)-dimensional
k-constrained modified KP hierarchy):
\begin{equation}\label{2+1L1M3}
\begin{array}{l}
L_0=\partial_y-{\bf q}{\cal M}_0D^{-1}{\bf r}^{\top}D,\,\, \alpha\in{\mathbb{C}},\\
M_3=\alpha_3\partial_{t_3}+c_0D^3-c_2\partial_y^3- 3c_0v_1D^2-
3c_0v_3D+3c_2{\bf q}_y{\cal M}_0D^{-1}\partial_y{\bf
r}^{\top}D+\\+c_2{\bf q}{\cal M}_0D^{-1}\partial_y{\bf
r}^{\top}\partial_yD-3c_2{\bf q}{\cal M}_0\partial_yD^{-1}{\bf
r}^{\top}D{\bf q}{\cal M}_0D^{-1}{\bf r}^{\top}D+\\+3c_2{\bf q}{\cal
M}_0D^{-1}\{{\bf r}^{\top}{\bf q}_x\}_yD^{-1}{\cal M}_0{\bf
r}^{\top}D,\, c_0,c_2\in{\mathbb{R}}.
\end{array}
\end{equation}
where ${\bf q}$ and ${\bf r}$ are $1\times m$-vector functions, $u$,
$v_1$ and $v_3$ are scalar functions, ${\cal M}_0$ is an $m\times
m$-constant matrix. If we impose additional reduction in
(\ref{2+1L1M3}):
\begin{equation}\label{RED}
\begin{array}{l}
{\bf
q}=(\alpha_1q_1,\alpha_1q_2,\alpha_2D^{-1}\{u\},\alpha_2),\,\,{\bf
r}={\bf q},\,{\cal
M}_0={\rm{diag}}(I_2,\sigma),\,\,\\\sigma=\left(\begin{array}{cc}
0&1\\-1&0\end{array}\right), \alpha_1,\alpha_2\in{\mathbb{R}},
\end{array}
\end{equation}
where $I_2$ is the identity matrix of dimension $2\times 2$; $q_1$,
$q_2$ and $u$ are scalar functions. The system that is equivalent to
the Lax equation $[L_0,M_3]=0$ generalizes (2+1)-extension of
modified KdV equation that was investigated in \cite{PHD,LW}
(reduction $c_2=0$, $\alpha_2=0$ in (\ref{2+1L1M3})-(\ref{RED})) and
the Nizhnik equation (\ref{DSL1M3r2scal}) (reduction $\alpha_1=0$ in
(\ref{RED})).

The following remark shows connections between constrained KP
hierarchies and Lax pairs with recursion operators.
\begin{remark}\label{rem11}
It should be also mentioned that the majority of Lax representations
with recursion operators (the representations connected with
bi-hamiltonian pairs of the corresponding integrable systems)
involve integro-differential operators \cite{VSG}. A part of such
representations is contained in the constrained KP hierarchies
(presented in Section III and IV) with additional (not only a
Hermitian conjugation reduction) specific reductions.
\end{remark}
 From Remark \ref{rem11} it follows that one of the essential
problems for further investigation is an adaptation of the dressing
method (Theorem \ref{2009} and \ref{M} from Section IV) for Lax
representations with the recursion operators. The work in this
direction is in progress.
\begin{remark}\label{rem12}
It should be also pointed out that the interesting problem consists
in the possibility of obtaining other classes of exact solutions
(e.g., lumps and dromions; see Remark \ref{RemDrom}) for nonlinear
equations from the proposed (2+1)-BDk-cKP hierarchy.
\end{remark}
We plan to consider the above mentioned problem (Remark \ref{rem12})
in a
separate paper.

\section{Acknowledgment}
The authors thank Professor M\"uller-Hoissen for fruitful
discussions and useful advice in preparation of this paper. The
authors also wish to express their gratitude to the reviewer for the
valuable comments and suggestions and drawing their attention to
Refs. \cite{Gerdjikov1,Gerdjikov2} and \cite{BLMP}.

 O.I. Chvartatskyi thanks to German Academic Exchange Service (DAAD) for financial
support (Codenumber A/12/85461). Yu.M.~Sydorenko (J. Sidorenko till
1998 and Yu.M.~Sidorenko till 2002 in earlier transliteration) is
grateful to the Ministry of Education, Science, Youth and Sports of
Ukraine for partial financial support (Research Grant MA-107F).

\section{Appendix. Proof of Lemma 1.}
\begin{proof}
We will use the following recurrent formulae that can easily be
checked by direct calculation:
\begin{equation}\label{mw}
{\cal M}_{2}=C\Lambda C^{-1}{\cal M}_1+{\cal M}_1
C\tilde{\Lambda}^{\top}C^{-1},
\end{equation}
\begin{equation}\label{vc}
{\cal M}_{l+1}=C\Lambda C^{-1}{\cal M}_{l}+{\cal
M}_{l}C^{-1}{\tilde{\Lambda}}^{\top}C-C\Lambda C^{-1}{\cal
M}_{l-1}C^{-1}{\tilde{\Lambda}}^{\top}C.
\end{equation}
At first we will prove the following formula:
\begin{equation}\label{fla}
{\cal M}_{l+1}\!\!=\!\!\sum_{s=0}^kC\Lambda^sC^{-1}{\cal
M}_{l-k+1}C^{-1}({\tilde{\Lambda}}^{\top})^{k-s}\!C\!-\!\sum_{s=1}^kC\Lambda^sC^{-1}{\cal
M}_{l-k}C^{-1}({\tilde{\Lambda}}^{\top})^{k-s+1}C,k\leq l\!-\!2,
\end{equation}
via induction by $k$. Assume that (\ref{fla}) holds for some
$k<l-2$. We shall show that then this formula holds for $k+1<l-1$
using (\ref{vc}):
\begin{equation}
\begin{array}{l}
{\cal M}_{l+1}=\sum\limits_{s=0}^kC\Lambda^s C^{-1}{\cal
M}_{l-k+1}C^{-1}(\tilde{\Lambda}^{\top})^{k-s}C-\sum\limits_{s=1}^kC\Lambda^s
C^{-1}{\cal
M}_{l-k}C^{-1}(\tilde{\Lambda}^{\top})^{k-s+1}C=\\=\sum\limits_{s=0}^k
C\Lambda^{s+1} C^{-1}{\cal
M}_{l-k}C^{-1}(\tilde{\Lambda}^{\top})^{k-s}C+\sum\limits_{s=0}^kC\Lambda^s
C^{-1}{\cal
M}_{l-k}C^{-1}(\tilde{\Lambda}^{\top})^{k-s+1}C-\\-\sum\limits_{s=0}^kC\Lambda^{s+1}
C^{-1}{\cal
M}_{l-k-1}C^{-1}(\tilde{\Lambda}^{\top})^{k-s+1}C-\sum\limits_{s=1}^k
C\Lambda^s C^{-1}{\cal
M}_{l-k}C^{-1}(\tilde{\Lambda}^{\top})^{k+1-s}C=\\
=\sum\limits_{i=1}^kC\Lambda^i C^{-1}{\cal
M}_{l-k}C^{-1}(\tilde{\Lambda}^{\top})^{k+1-i}C+ {\cal
M}_{l-k}C^{-1}(\tilde{\Lambda}^{\top})^{k+1}C-\\-\sum\limits_{i=1}^{k+1}
C\Lambda^i C^{-1}{\cal
M}_{l-k-1}C^{-1}(\tilde{\Lambda}^{\top})^{k-i+2}C=
\sum\limits_{i=0}^{k+1}C\Lambda^i C^{-1}{\cal
M}_{l-k}C^{-1}(\tilde{\Lambda}^{\top})^{k-i+1}C-\\-\sum\limits_{i=1}^{k+1}
C\Lambda^i C^{-1}{\cal
M}_{l-k-1}C^{-1}(\tilde{\Lambda}^{\top})^{k-i+2}.
\end{array}
\end{equation}
Thus, we have proven (\ref{fla}). After the substitution of $k=l-2$
in (\ref{lemma11}) and using (\ref{mw}) we obtain:
\begin{equation}
\begin{array}{l}
{\cal M}_{l+1}=\sum\limits_{s=0}^{l-1}C\Lambda^sC^{-1}{\cal
M}_{2}C^{-1}({\tilde{\Lambda}}^{\top})^{l-1-s}-\sum\limits_{s=1}^{l-1}C\Lambda^sC^{-1}{\cal
M}_{1}C^{-1}({\tilde{\Lambda}}^{\top})^{l-s}=\\=\sum\limits_{s=0}^{l-1}C\Lambda^{s+1}C^{-1}{\cal
M}_{1}C^{-1}({\tilde{\Lambda}}^{\top})^{l-1-s}+\sum\limits_{s=0}^{l-1}C\Lambda^sC^{-1}{\cal
M}_{1}C^{-1}({\tilde{\Lambda}}^{\top})^{l-s}C-\\-\sum\limits_{s=1}^{l-1}C\Lambda^{s}C^{-1}{\cal
M}_{1}C^{-1}({\tilde{\Lambda}}^{\top})^{l-s}C=\sum\limits_{s=0}^{l}C\Lambda^sC^{-1}{\cal
M}_{1}C^{-1}({\tilde{\Lambda}}^{\top})^{l-s}C.\end{array}
\end{equation}
This finishes the proof of formula (\ref{lemma11}) and Lemma 1.

\end{proof}


\end{document}